\documentclass[12pt]{article}
\usepackage{amsmath}
\usepackage{graphicx}
\usepackage{enumerate}
\usepackage{natbib}
\usepackage{url} 

\usepackage{hyperref}

\hypersetup{
  colorlinks,
  citecolor=black,
  linkcolor=black,
  urlcolor=blue}

\usepackage{amsmath}
\usepackage{graphicx,amssymb,booktabs,multirow}
\usepackage{anysize}
\usepackage{bm}
\usepackage{natbib,amsmath}
\usepackage{graphicx}
\usepackage{fancyhdr}
\usepackage{rotating}

\makeatletter
\newcommand*{\indep}{%
	\mathbin{%
		\mathpalette{\@indep}{}%
	}%
}

\newcommand*{\nindep}{%
	\mathbin{
		\mathpalette{\@indep}{\not}
	}%
}
\newcommand*{\@indep}[2]{%
	\sbox0{$#1\perp\m@th$}
	\sbox2{$#1=$}
	\sbox4{$#1\vcenter{}$}
	\rlap{\copy0}
	\dimen@=\dimexpr\ht2-\ht4-.2pt\relax
	\kern\dimen@
	{#2}%
	\kern\dimen@
	\copy0 
}
\makeatother

\usepackage{amsthm}

\makeatother
\newtheorem{theorem}{Theorem}
\newtheorem{lemma}{Lemma}

\newtheorem{proposition}{Proposition}
\newtheorem{assumption}{Assumption}

\usepackage{bbm}

\usepackage{caption}
\usepackage{subcaption}

\def\var{\textnormal{Var}}

\usepackage[all]{xy}
\usepackage{bm}
\usepackage{pifont}
\usepackage{stmaryrd}
\usepackage{pgf,tikz}
\usepackage{tikz}
\usetikzlibrary{shapes,arrows,positioning,calc}

\usepackage{mathrsfs}

\newcommand\Cone{{\mathcal{C}_1}}
\newcommand\Ctwo{{\mathcal{C}_2}}

\makeatletter
\newcommand{\mathleft}{\@fleqntrue\@mathmargin0pt}
\newcommand{\mathcenter}{\@fleqnfalse}
\makeatother

\usepackage{algorithm}
\usepackage{algorithmic}
\usepackage{wrapfig}

\allowdisplaybreaks

\usepackage{accents}

\def\T{{ \mathrm{\scriptscriptstyle T} }}

\newcommand\MR{{\textnormal{CAL}}}
\newcommand\DR{{\textnormal{DR}}}

\newcommand\Ber{{\textnormal{Ber}}}
\newcommand\logit{{\textnormal{logit}}}

\newcommand\lexpd{\text{\text{lexpd}}}
\newcommand\lnetw{\text{\text{lnetw}}}
\newcommand\lincome{\text{\text{lincome}}}
\newcommand\age{\text{\text{age}}}
\newcommand\sex{\text{\text{sex}}}
\newcommand\marital{\text{\text{marital}}}
\newcommand\eduone{\text{\text{edu1}}}
\newcommand\edutwo{\text{\text{edu2}}}
\newcommand\eduthr{\text{\text{edu3}}}
\newcommand\white{\text{\text{white}}}
\newcommand\black{\text{\text{black}}}

\begin{document}

	\def\spacingset#1{\renewcommand{\baselinestretch}%
		{#1}\small\normalsize} \spacingset{1}

	
		\title{\bf Calibrated regression estimation using empirical likelihood  under data fusion}
 		\author{Wei Li$^1$, Shanshan Luo$^2$, and Wangli Xu$^1$	\vspace{0.5cm}
 		\\
 		$^1$Center for Applied Statistics and School of Statistics, \\
 		Renmin University of China\vspace{0.1cm}
 		\\
 		$^2$School of Mathematical Sciences, Peking University}
		\date{}
		\maketitle
	
		
		\bigskip
		\begin{abstract}
		Data analysis based on information  from several sources is common in economic and biomedical studies.
 This setting is often referred to as the data fusion problem, which differs from traditional missing data problems since no complete data is observed for any subject.
 We consider a regression  analysis when the outcome variable and some  covariates are collected from two different sources. By leveraging the common variables observed in both data sets, doubly robust estimation procedures  are proposed in the literature to protect against possible model misspecifications. However, they employ only a single propensity score model for the data fusion process and a single imputation model for the  covariates available in one data set. It may be  questionable to assume that either model is correctly specified in practice. We therefore propose an approach that calibrates multiple propensity score  and  imputation models to gain more protection based on empirical likelihood methods. The resulting estimator is consistent when any one of those models is correctly specified and is robust against extreme values of the fitted propensity scores. We also establish its asymptotic normality property  and discuss the semiparametric estimation efficiency. Simulation studies show that the proposed estimator has substantial advantages over  existing doubly robust estimators, and an assembled  U.S. household expenditure data example is used for illustration.  
		\end{abstract}
		
		\noindent%
		\vfill
		
		\newpage
		\spacingset{1.8} 
	\section{Introduction}

	There are many situations in practice where all relevant variables for addressing particular scientific hypotheses  cannot be collected from a single data source. For example, in a study investigating effects of housing project participation on living qualities of poor families, 
the Current Population Survey data consists of project participation and the Census data consists of living quality attributes \citep{currie2000public,shu2020improved}.
Another example studies the relationship between households' consumption and saving behavior, where information on consumption is available in the Panel Study of Income Dynamics, and information on  wealth is obtained from separate data sources, e.g., Health and Retirement Survey  \citep{blundell2008consumption,buchinsky2022estimation}. Two-sample instrumental variable analysis in causal inference literature is an additional example with important variables collected separately, and it includes two-sample Mendelian randomization as one of the most exciting applications in genetic epidemiology \citep{angrist1992effect,pierce2013efficient}. 
It is natural to try to combine the two data sets to answer specific research questions in these scenarios. Analysis of such combined data is challenging, because no subject belongs to both sources and some variables available in one source are never observed in the other one. This setting is typically referred to as the data fusion problem \citep{d2006statistical,ridder2007econometrics,rassler2012statistical}. 
Due to violations of the positivity assumption for any subject, data fusion problems are distinct from most missing-data problems and require specific analysis techniques.



In this article, we consider the setting where a primary sample provides measurements of the outcome $Y$ and covariates $V$,  while an auxiliary sample contains measurements of $V$ and additional covariates $W$. That is, the outcome $Y$ is collected only in the primary data, a subset of covariates $W$ is observed only in the auxiliary data, and the common variables $V$ are available in both data sets.
We are interested in regression analysis of the outcome given all covariates, and our aim is to estimate the $p$-dimensional parameter $\theta_0=(\theta_{10}^\T,\theta_{20}^\T)^\T$ defined through $E(Y\mid W,V)=\mu(\theta_{10}^\T W+\theta_{20}^\T V)$, where $\mu(\cdot)$ is some known monotone and continuously differentiable function. If a random sample of all variables is available, then consistent estimation of $\theta_0$ under regularity conditions is straightforward by solving a set of properly-chosen estimating equations. However,  when two data sets from separate sources are fused together,
the parameter of interest  $\theta_0$  cannot even be simply identified  from the observed data without additional conditions.
A prominent strategy in existing literature is to  impose one of the following conditional independence assumptions: $Y\indep W\mid V$ or $Y\indep V\mid W$ \citep{ogburn2021warning,miao2021invited}. The former assumption is   fundamental
in statistical matching literature \citep{d2006statistical,rassler2012statistical,hirukawa2018consistent}, where the variables $Y$ and $W$ collected from separate samples are matched by the common variables $V$. The latter is analogous to exclusion restrictions in instrumental variable analysis, which forms the basis for the validity of two-sample instrumental variable estimators \citep{angrist1992effect,graham2016efficient}. However, these two assumptions are problematic in our setting if a potential non-null association between $Y$ and $W$ (or $V$) is the scientific hypothesis under consideration. We clearly pinpoint
more reasonable  conditions under which $\theta_0$ is identifiable in Section~\ref{sec:identification}.

When model identifiability has been guaranteed, a large class of  semiparametric and parametric estimation methods for fused data sets is proposed. \cite{chen2008semiparametric} utilized nonparametric series estimation on a propensity score model for the data source process or an imputation model for the partly observed covariates to achieve consistent results in separable moment restriction models. However, such methods hinge on certain smoothness conditions that are often assumed and can be problematic with a high-dimensional vector of the common variables $V$. Recently, \cite{graham2016efficient} proposed a more flexibly parametric modeling approach, which is doubly robust in that it is consistent if the  propensity score  or the imputation model is correctly specified in a certain class of nuisance models.
To break this limitation, \cite{shu2020improved} and \cite{evans2021doubly}  developed general doubly robust estimators with unrestricted nuisance model specifications. Doubly robust procedures offer some protection against model misspecification, but as only a single propensity score model  and a single imputation model are allowed, they may not provide sufficient protection and it is restrictive to assume that either model is correctly specified in practice. In addition, since these doubly robust procedures are based on inverse probability weighting, they may suffer from large variances due to extreme probability values.

We propose a calibration procedure 
that  allows multiple models for the propensity score and imputation models to gain more protection in this paper.
It is appealing to fit multiple models in data fusion problems, especially when there are a large number of common variables, since a  correct specification for the nuisance models in such a case is difficult.  Although a variety of variable selection techniques can be employed for model building,  their performances typically rely on levels of tuning parameters, and different tuning parameters may lead to different models. Selecting tuning parameters based on some information criterion would be an option, but the selection itself brings additional uncertainty in the working model specifications. It thus seems desirable to postulate a set of reasonable models, each involving different subsets of covariates and possibly different link functions, and to incorporate them simultaneously. 
While the idea of fitting multiple models has been well developed in missing data literature \citep{han2013estimation,han2014multiply,chan2014oracle,chen2017multiply,li2020demystifying}, it remains unclutivated in the area of
data fusion studies. The implementation of this idea in such studies
becomes more complicated and requires additional techniques.		
For example,	\cite{han2014multiply} considered regression analysis with missing outcome data, in which only one calibration weight is needed to incorporate multiple models.  The calibration weight is constructed by estimating the conditional expecation of the outcome given all covariates that can be easily implemented. However, such a strategy cannot be directly applied to data fusion problems since some covariates are not fully observed. We thus employ an imputation approach to obtain approximations of the conditional expectation of the outcome given only always-observed covariates, and use these approximations to construct two different calibration weights that are  needed for
data fusion analysis. This leads to complexity in both
the implementation and theoretical investigations, where some additional empirical process theories are required for the asymptotic results of the proposed estimator.
Different from existing doubly robust procedures for data fusion problems, our  estimation strategy relies on the empirical likelihood method \citep{owen2001empirical,qin1994empirical,qin2007empirical}, which circumvents the use of inverse probabilities. The resulting estimator is robust against extreme values of the fitted propensity scores. Furthermore,
the proposed estimator enjoys well-established theoretical properties. Under some regularity conditions, it
is consistent and asymptotically normal if any one of the propensity score and imputation models is correct.  The estimator also attains semiparametric efficiency bound when both one of  propensity scores and one of imputation models are correctly specified, without requiring knowledge of which two models are correct.

The remainder of this paper is organized as follows. In Section~\ref{sec:identification}, we introduce notation and discuss identifiability of the parameter of interest. In Section~\ref{sec:estimator}, we provide a calibration-based estimation approach. We then establish the asymptotic results for the proposed estimator in Section~\ref{sec:asymptotic}. We study the
finite-sample performance of the proposed approach via both simulation studies and an assembled U.S. household expenditure  data example in Sections~\ref{sec:sim} and~\ref{sec:application}, respectively. We conclude with a discussion in Section~\ref{sec:discussion} and relegate proofs to
the Appendix. 

	\section{Notation, assumptions and identifiability}\label{sec:identification}

Suppose there are $n$ individuals who are merged from two different sources.
Let $R$ denote the data source indicator with $R=1$ if a subject is observed in the primary sample and $R=0$ if it is observed in the auxiliary sample. Let $m=\sum_{i=1}^n R_i$ be the number of subjects who are observed in the primary sample, and index those subjects by $i=1,\ldots,m$ without loss of generality.
The outcome  $Y$ is only available in the primary sample, some covariates $W$ are only available in the auxiliary sample, and the other covariates $V$ are observed in both data sources.  Let $f(\cdot)$ denote the probability density or mass function of a random variable (vector), and let $\pi(v)=f(R=1\mid V=v)$ denote the probability that a subject is observed in the primary sample given the common variables $V=v$.
We make the following assumptions.



\begin{assumption}\label{ass:ignorability}
	$R\indep (Y,W)\mid V$.
\end{assumption}

\begin{assumption}\label{ass:positivity}
	$\delta<\pi(v)<1-\delta$ for all $v$ and some fixed constant $\delta\in(0,1)$.
\end{assumption}

Assumption~\ref{ass:ignorability} is similar to missing at random in missing data problems. It implies that the  conditional distribution of $Y$ or $W$  given the always-measured variables $V$  does not vary across the primary and auxiliary populations, but allows the marginal distributions of $V$ to differ between these two populations. Assumption~\ref{ass:positivity}  indicates that for each subject in one data source, the probability of observing a matching unit in the other data source with similar values of $V$ is positive.
It is different from the usual positivity assumption in missing data problems which requires a positive chance of observing complete data for each subject. Both assumptions are basic and frequently made in data fusion problems \citep{chen2008semiparametric}. 	
The identifiability of the parameter  $\theta_0$ in our regression problem can be guaranteed if the conditional distribution $f(Y\mid W,V)$ is identified. For identification of this conditional distribution,  we need further assumptions.

\begin{assumption}\label{ass:instrumental-variable}
	There exists a variable $Z$ in the always-observed covariates $V$ such that	$Y\indep Z\mid W,X$, where $X$  denotes the remaining covariates in $V$.
\end{assumption}
\begin{assumption}\label{ass:completeness-condition}
	For any function $\phi(W,X)$ with finite mean, $E\{\phi(W,X)\mid Z,X\}=0$ implies $\phi(W,X)=0$ almost surely.
\end{assumption}

For convenience,  we may use the notation $V$ and $(Z,X^\T)^\T$ interchangeably below.
Assumption~\ref{ass:instrumental-variable} reveals that
the scalar covariate $Z$ affects the outcome $Y$ only through its association with $W$ and $X$. In contrast to the commonly-used exclusion restriction assumption $Y\indep V\mid W$, this assumption requires only one of the covariates in $V$ to be conditionally independent of $Y$, which is weaker and more reasonable in practice. \cite{ridder2007econometrics} proposed similar assumptions and discussed identification in the context of categorical variables. To achieve a general identification result, we impose the completeness condition in Assumption~\ref{ass:completeness-condition}. This condition is widely used for model
identifiability across various disciplines   and can be satisfied for many commonly-used parametric models \citep{newey2003instrumental,d2010new}.


\begin{proposition}\label{prop:identifiable}
	Under Assumptions~\ref{ass:ignorability}--\ref{ass:completeness-condition}, the conditional distribution $f(Y\mid W,V)$ is identifiable.
\end{proposition}

The identifiability of $f(Y\mid W,V)$ in Proposition~\ref{prop:identifiable} implies that the regression parameter $\theta_0$ is identifiable under Assumptions~\ref{ass:ignorability}--\ref{ass:completeness-condition}.
When Assumptions~\ref{ass:instrumental-variable} and~\ref{ass:completeness-condition} are removed, one may also achieve (local) identifiability of $\theta_0$ based on the moment condition:
\begin{align}\label{eqn:estimatingeq}
	E\Big[ \big\{Y-\mu(\theta_{1}^\T W+\theta_{2}^\T V)\big\}t(V;\theta)\Big]=0,
\end{align}
where $t(V;\theta)$ is a  user-specified $p$-dimensional vector function that may depend on $V$ and $\theta$. Denote $\Gamma(\theta)$ to be the $p\times p$ derivative matrix of the left-hand side of the above equation, i.e., $\Gamma(\theta) = E\{Y\partial t(V;\theta)/\partial\theta-\partial s(W,V;\theta)/\partial\theta\}$,
where $s(W,V;\theta)=\mu(\theta_{1}^\T W+\theta_{2}^\T V)t(V;\theta)$. We simply write $\Gamma = \Gamma(\theta_0)$.
If $\Gamma$ is of full rank, then $\theta_0$ is (locally) identifiable \citep{bandeen1997latent}. Under some scenarios, the global identifiability of $\theta_0$ can be guaranteed. For example, 
when $\mu(\cdot)$ is the identity function, $\theta_0$ is identifiable if there exists a nonlinear term of $V$ in $E(W\mid V)$ \citep{evans2021doubly,miao2021invited}. In the next section, we assume $\theta_0$ has been identified and propose a procedure to estimate $\theta_0$ based on the moment condition~\eqref{eqn:estimatingeq} with a fixed function $t(V;\theta)$.

	\section{Proposed estimator}\label{sec:estimator} 
Let $\Cone=\{\pi^j(\eta^j):j=1,\ldots J\}$ be a set of $J$ propensity score models for $\pi(V)$
and $\Ctwo=\{a^k(\gamma^k):k=1,\ldots,K\}$ be a set of $K$ imputation models for $f(W\mid V)$. Here $\eta^j$ and $\gamma^k$ are the corresponding parameters, and we let $\widehat\eta^j$ and $\widehat\gamma^k$ denote their estimators. Usually, each $\widehat\eta^j$ is obtained by maximizing the binomial likelihood function:
\begin{align*}
	\prod_{i=1}^n \big\{\pi_i^j(\eta^j)\big\}^{R_i}\big\{1-\pi_i^j(\eta^j)\big\}^{1-R_i},
\end{align*} 
and each $\widehat\gamma^k$ is similarly obtained by maximum likelihood estimation based on auxiliary data.
Let $\widehat\theta^k$ denote the solution to
\begin{equation}\label{eqn:estimator-theta-k}
	\begin{aligned}
		&\frac{1}{n}\sum_{i=1}^n t(V_i;\theta)\Big[R_i\big\{Y_i-E(Y\mid V_i;\widehat\gamma^k,\theta)\big\}
		\\&~~~~~~~~~~~~~~~~~~~~~~~~~ 
		+(1-R_i)\big\{E(Y\mid V_i;\widehat\gamma^k,\theta)-\mu(\theta_1^\T W_i+\theta_2^\T V_i) \big\}\Big]=0.
	\end{aligned}
\end{equation}
It is easy to verify that if the $k$th imputation model is correctly specified, then $\widehat\theta^k$ is a consistent estimator of $\theta_0$.
Let $\{W_i^d(\widehat\gamma^k):d=1,\ldots,D\}$ denote $D$ random draws from $f(W\mid V_i;\widehat\gamma^k)$, and define  $g_{i}(\widehat\gamma^k,\widehat\theta^k)=D^{-1}\sum_{d=1}^D s\{W_i^d(\widehat\gamma^k),V_i;\widehat\theta^k\}$ for $i=1,\ldots,n$, where $s(W,V;\theta)$ is defined below~\eqref{eqn:estimatingeq}. Because  $ E(Y\mid V)= E\{E(Y\mid W,V)\mid V\}$, 
the quantity $g_{i}(\widehat\gamma^k,\widehat\theta^k)$ can be seen as an estimate of $E(Y\mid V_i;\widehat\gamma^k,\widehat\theta^k)t(V_i;\widehat\theta^k)$ by averaging over the $D$ random draws taken from $f(W\mid V_i;\widehat\gamma^k)$.

Our procedure consists of three steps. In the first step, we obtain the calibration weights $\omega_{1i}$ 
for subjects in the primary sample $\{i:i=1,\ldots,m\}$ by imposing the following  constraints:
\begin{equation}\label{eqn:omega1-constraints}
	\begin{aligned}
		&\omega_{1i}\geq 0,\quad \sum_{i=1}^m\omega_{1i}=1,\quad
		\sum_{i=1}^m \omega_{1i} \pi_i^j(\widehat\eta^j)=\widehat\tau^j\qquad (j=1,\ldots,J),\\
		&~~~~~~~~~~~~~~~~~~\sum_{i=1}^m \omega_{1i}g_{i}(\widehat\gamma^k,\widehat\theta^k)=\widehat\psi^k\qquad (k=1,\ldots,K),
	\end{aligned}
\end{equation}
where
$\widehat\tau^j=n^{-1}\sum_{i=1}^n \pi_i^j(\widehat\eta^j)$
and $\widehat \psi^k = n^{-1}\sum_{i=1}^n g_{i}(\widehat\gamma^k,\widehat\theta^k)$. The rationale behind these constraints is as follows. Note that for any function $b(V)$ with finite expectation, we have 
\begin{equation}\label{eqn:population-omega1-constraints}
	\begin{aligned}	
			E\big[\omega(V)\big\{b(V)-E(b(V))\big\}\mid R=1 \big]=0,
	\end{aligned}
\end{equation}
where $\omega(V)=1/\pi(V)$. We then take $b(V)$ to be $\pi^j(\widehat\eta^j)$ and $E(Y\mid V;\widehat\gamma^k,\widehat\theta^k)t(V;\widehat\theta^k)$
to obtain the above constraints.
We choose $\widehat\omega_{1i}$'s that maximize $\prod_{i=1}^m \omega_{1i}$ subject to the constraints in~\eqref{eqn:omega1-constraints}. To give an explicit form for the estimates $\widehat\omega_{1i}$'s, we write
\begin{align*}
	&~~~~~~~\widehat\eta = \big\{(\widehat\eta^1)^\T,\ldots,(\widehat\eta^J)^\T\big\}^\T,~~\widehat\gamma = \big\{(\widehat\gamma^1)^\T,\ldots,(\widehat\gamma^K)^\T\big\}^\T,~~\widehat\theta =\big \{(\widehat\theta^1)^\T,\ldots,(\widehat\theta^K)\big \}^\T,\\
	&\widehat h_{i}(\widehat\eta,\widehat\gamma,\widehat\theta) = \big[\pi_i^1(\widehat\eta^1)-\widehat\tau^1,\ldots,\widehat\pi_i^J(\widehat\eta^J)-\widehat\tau^J,
	\{g_{i}(\widehat\gamma^1,\widehat\theta^1)-\widehat\psi^1\}^\T,\ldots,
	\{ g_{i}(\widehat\gamma^K,\widehat\theta^K)-\widehat\psi^K\}^\T \big]^\T.
\end{align*}
Then by Lagrange multipliers method, we have
\begin{equation*}\label{eqn:expression-for-omega1}
	\widehat\omega_{1i} = \frac{1}{m}\frac{1}{1+\widehat\rho^\T\widehat h_{i}(\widehat\eta,\widehat\gamma,\widehat\theta)}\bigg/\bigg\{\frac{1}{m}\sum_{i=1}^m\frac{1}{1+\widehat\rho^\T\widehat h_{i}(\widehat\eta,\widehat\gamma,\widehat\theta)} \bigg\}\qquad (i=1,\ldots,m),
\end{equation*}
where $\widehat\rho = (\widehat\rho_1,\ldots,\widehat\rho_{J+pK})^\T$ is  a $(J+pK)$-dimensional vector satisfying the equation
\begin{equation*}\label{eqn:lagrange-multiplier-omega1}
	\sum_{i=1}^m\frac{\widehat h_{i}(\widehat\eta,\widehat\gamma,\widehat\theta)}{1+\widehat\rho^\T \widehat h_{i}(\widehat\eta,\widehat\gamma,\widehat\theta)}=0.
\end{equation*}
To guarantee nonnegativity of $\widehat\omega_{1i}$, we further impose the condition
$1+\widehat\rho^\T \widehat h_{i}(\widehat\eta,\widehat\gamma,\widehat\theta)> 0$ for $i=1,\ldots,m$. Then under this condition, we can obtain $\widehat\rho$ by minimizing a convex function $F(\rho)=n^{-1}\sum_{i=1}^n R_i\log\{1+\rho^\T\widehat h_{i}(\widehat\eta,\widehat\gamma,\widehat\theta)\}$.

In the second step, we obtain the weights $\omega_{0i}$'s for subjects in the auxiliary data $\{i:i=m+1,\ldots,n\}$ similarly by the following constraints:
\begin{equation}\label{eqn:omega0-constraints}
	\begin{aligned}
		&\omega_{0i}\geq 0,\quad \sum_{i=m+1}^n\omega_{0i}=1,\quad
		\sum_{i=m+1}^n \omega_{0i} \pi_i^j(\widehat\eta^j)=\widehat\tau^j\qquad (j=1,\ldots,J),\\
		&~~~~~~~~~~~~~~~~~~\sum_{i=m+1}^n \omega_{0i}g_{i}(\widehat\gamma^k,\widehat\theta^k)=\widehat\psi^k\qquad (k=1,\ldots,K),
	\end{aligned}
\end{equation}
Then the estimates $\widehat\omega_{0i}$'s that maximize $\prod_{i=m+1}^n\omega_{0i}$ under constraints in~\eqref{eqn:omega0-constraints}  are given by
\begin{equation*}\label{eqn:expression-for-omega0}
	\widehat\omega_{0i}=\frac{1}{n-m}\frac{1}{1+\widehat\alpha^\T\widehat h_{i}(\widehat\eta,\widehat\gamma,\widehat\theta)}\bigg/\bigg\{\frac{1}{n-m}\sum_{i=m+1}^n\frac{1}{1+\widehat\alpha^\T\widehat h_{i}(\widehat\eta,\widehat\gamma,\widehat\theta)} \bigg\}~~~ (i=m+1,\ldots,n),
\end{equation*}
where $\widehat\alpha = (\widehat\alpha_1,\ldots,\widehat\alpha_{J+pK})^\T$ is the $(J+pK)$-dimensional Lagrange multipliers solving
\begin{equation*}\label{eqn:lagrange-multiplier-omega0}
	\begin{aligned}
		\sum_{i=m+1}^n\frac{\widehat h_{i}(\widehat\eta,\widehat\gamma,\widehat\theta)}{1+\widehat\alpha^\T\widehat h_{i}(\widehat\eta,\widehat\gamma,\widehat\theta)}=0, ~~~ \text{and}~~~ 1+\widehat\alpha^\T\widehat h_{i}(\widehat\eta,\widehat\gamma,\widehat\theta)> 0~~~(i=m+1,\ldots,n).
	\end{aligned}
\end{equation*}

Finally, the proposed estimator of $\theta_0$ based on the calibration weights, denoted by $\widehat\theta_{\MR}$, is the solution to
\begin{equation*}\label{eqn:mr-estimator}
	\sum_{i=1}^m \widehat\omega_{1i} Y_it(V_i;\theta)-\sum_{i=m+1}^n \widehat\omega_{0i} s( W_i,V_i;\theta)=0.
\end{equation*}
Compared to existing doubly robust estimators that weight each subject in the primary data by $1/\{n\widehat\pi(V)\}$ and  subject in the auxiliary data by $1/[n\{1-\widehat\pi(V)\}]$, the proposed calibration estimator $\widehat\theta_{\MR}$ use weights $\widehat\omega_{1i}$ and $\widehat\omega_{0i}$, respectively. Those doubly robust estimators may be sensitive to near-zero or near-one values of $\widehat\pi(V)$, which can yield extremely large weights that may make the numerical performance be quite unstable. In our procedure, we obtain the calibration weights through maximization of $\prod_{i=1}^m \omega_{1i}$ and $\prod_{i=m+1}^n\omega_{0i}$ that satisfy~\eqref{eqn:omega1-constraints} and~\eqref{eqn:omega0-constraints}, respectively. These two objective functions increase if their separate weights are more evenly distributed, because the weights are restricted to be
nonnegative and sum-to-one. Thus, our procedure will not be affected dramatically  by extreme values of propensity score and 
should lead to more stable performance. This property inherits from the empirical likelihood method, which has been successfully applied to address missing data problems \citep{qin2007empirical,qin2009empirical,cao2009improving,han2014multiply,han2019general}.

\section{Asymptotic results}\label{sec:asymptotic}

In this section, we show that the estimator $\widehat\theta_{\MR}$ is consistent when one of the propensity scores or one of the imputation models is correctly specified. We also establish its asymptotic normality property and discuss the estimation efficiency. We introduce more notations that will be used later.    
Let $\eta_*^j$, $\gamma_*^k$, $\theta_*^k$,  $\tau_*^j$ and $\psi_*^k$ denote the probability limits of $\widehat\eta^j$, $\widehat\gamma^k$, $\widehat\theta^k$, $\widehat\tau^j$ and $\widehat \psi^k$ respectively, as $n\rightarrow\infty$, where $j=1,\ldots,J$ and $k=1,\ldots,K$. It is clear that  $\tau_*^j=E\{\pi^j(\eta_*^j)\}$ and $\psi_*^k=E[s\{W^d(\gamma_*^k),V;\theta_*^k\}]$. Write $\eta_*^\T=\{(\eta_*^1)^\T,\ldots,(\eta_*^J)^\T\}$, $\gamma_*^\T=\{(\gamma_*^1)^\T,\ldots,(\gamma_*^K)^\T\}$, $\theta_*^\T=\{(\theta_*^1)^\T,\ldots,(\theta_*^K)^\T\}$, and
$	h_i(\eta_*,\gamma_*,\theta_*)^\T=[\pi_i^1(\eta_*^1)-\tau_*^1,\ldots,\pi_i^J(\eta_*^J)-\tau_*^J,\{g_i(\gamma_*^1,\theta_*^1)-\psi_*^1\}^\T,
\ldots,\{g_i(\gamma_*^K,\theta_*^K)-\psi_*^K\}^\T]$ for $i=1,\ldots,n$.

We first consider the case where one of the  models in $\Cone$ is correctly specified. Without loss of generality, let $\pi^1(\eta^1)$ be the correct model in the sense that there exists some value $\eta_0^1$ such that $\pi^1(\eta_0^1)=\pi(V)$.
To establish the consistency property of $ \widehat\theta_{\MR}$,
we  build the connection between $\widehat\omega_{1i}$ and another version of the empirical likelihood estimator obtained from the primary sample $\{i:i=1,\ldots,m\}$. Let $p_i$ denote the conditional empirical probability mass on $(Y_i,W_i,V_i)$ given $R_i=1$ for $i=1,\ldots,m$. According to~\eqref{eqn:population-omega1-constraints} and the fact that $\omega(V)=1/\pi^1(\eta_0^1)$, the estimator of $p_i$ is given by the following constrained optimization:
\begin{equation*}\label{eqn:p1-constraints}
	\begin{aligned}
		&\max_{p_1,\ldots,p_m}\prod_{i=1}^mp_i~~ \text{subject to}~~ p_i\geq 0,~
		\sum_{i=1}^m p_i \{\pi_i^j(\widehat\eta^j)-\widehat\tau^j\}/\pi_i^1(\widehat\eta^1)=0\quad (j=1,\ldots,J),\\
		&~~~~~~~~~~~~~~~~~~~~~~~~~~~~~~~~~~~~~~~~\sum_{i=1}^m p_i\{g_{i}(\widehat\gamma^k,\widehat\theta^k)-\widehat\psi^k\}/\pi_i^1(\widehat\eta^1)=0\quad (k=1,\ldots,K).
	\end{aligned}
\end{equation*}
Using the Lagrange multipliers again yields that
\begin{equation*}\label{expression-for-p1}
	\widehat p_i = \frac{1}{m}\frac{1}{1+\widehat\lambda^\T\widehat h_{i}(\widehat\eta,\widehat\gamma,\widehat\theta)/\pi_i^1(\widehat\eta^1)}\quad (i=1,\ldots,m),
\end{equation*}
where $\widehat\lambda = (\widehat\lambda_1,\ldots,\widehat\lambda_{J+pK})^\T$ is the $(J+pK)$-dimensional Lagrange multiplier satisfying
\begin{equation*}
	\sum_{i=1}^m\frac{\widehat h_{i}(\widehat\eta,\widehat\gamma,\widehat\theta)/\pi_i^1(\widehat\eta^1)}{1+\widehat\lambda^\T \widehat h_{i}(\widehat\eta,\widehat\gamma,\widehat\theta)/\pi_i^1(\widehat\eta^1)}=0.
\end{equation*}
With some simple algebra given in the supplementary material, one can show that
\begin{equation*}
	\widehat\omega_{1i} = \frac{1}{m}\frac{\widehat\tau^1/\pi_i^1(\widehat\eta^1)}{1+\widehat\lambda^\T \widehat h_{i}(\widehat\eta,\widehat\gamma,\widehat\theta)/\pi_i^1(\widehat\eta^1)}.
\end{equation*}
This equation links the calibration weight $\widehat\omega_{1i}$ with the Lagrange multiplier $\widehat\lambda$ of the empirical likelihood estimator $\widehat p_i$.
Based on the empirical likelihood theory, one can show that $\widehat\lambda=o_p(1)$. In addition, since $\pi^1(\eta^1)$ is a correctly-specified propensity score model, $\widehat\tau^1-m/n=o_p(1)$.  We then conclude from the above equation that $\widehat\omega_{1i}=1/\{n\pi_i^1(\widehat\eta^1)\}+o_p(1)$. 
Similarly, $\widehat\omega_{0i}=1/[n\{1-\pi_i^1(\widehat\eta^1)\}]+o_p(1)$. We would like to point out that these results do not contradict with our previous discussions that the proposed calibration weights are less sensitive to extreme propensity score values, because the discussions therein mainly emphasize the finite sample performance of those estimators. In an asymptotic way or when the sample size goes to infinity, the calibration weights are equivalent to the inverse propensity score weights. Based on the above intermediate results,  we obtain the following theorem.

\begin{theorem}\label{thm:consistency-propensity}
	When one of the models in $\Cone$ is correctly specified,  $\widehat\theta_{\MR}$ is a consistent estimator of $\theta_0$  as $n\rightarrow \infty$.
\end{theorem}

Next, we consider the case where one of the models in $\Ctwo$ is correct.
Without loss of generality, we assume that $a^1(\gamma^1)$ is correctly specified, i.e., $a^1(\gamma_0^1)=f(W\mid V)$ for some $\gamma_0^1$. Then we have $\gamma_0^1=\gamma_*^1$; that is, $\widehat\gamma^1$ is a consistent estimator of $\gamma_0^1$. This implies that the estimator $\widehat\theta^1$ obtained from~\eqref{eqn:estimator-theta-k} is a consistent estimator of $\theta_0$. By utilizing constraints in~\eqref{eqn:omega1-constraints} and \eqref{eqn:omega0-constraints}, we can also obtain the consistency of $\widehat\theta_{\MR}$ in this case. 

\begin{theorem}\label{thm:consistency-misscov}
	When one of the models in $\Ctwo$ is correctly specified, $\widehat\theta_{\MR}$ is  a consistent estimator of $\theta_0$ as $n\rightarrow \infty$.
\end{theorem}

Different from the consistency property of $\widehat\theta_{\MR}$, the derivation of the asymptotic distribution is asymmetric. In other words, the asymptotic normality property depends on
which of the $J+K$ candidate models is correctly specified, and the asymptotic variance is different when one of the propensity score models or one of the imputation models is correct. The usual strategy in traditional missing data analysis for semiparametric theory is to assume the propensity score model is correctly specified  \citep{
	robins1995semiparametric,tsiatis2006semiparametric}. We follow this way in the current data fusion setting.
Without loss of generality, we assume $\pi^1(\eta^1)$ is a correctly specified model for $\pi(V)$, and let $\Psi(\eta^1)$ denote the score function of $\eta^1$, that is,
\[
\Psi(\eta^1)=\frac{R-\pi^1(\eta^1)}{\pi^1(\eta^1)\{1-\pi^1(\eta^1)\}}\frac{\partial \pi^1(\eta^1)}{\partial\eta^1}.
\]
We simply write $\Psi=\Psi(\eta_0^1)$ and define the following matrices:
\begin{equation*}
	\begin{aligned}
		F=&E\bigg[\frac{Yt(V;\theta_0)-E\{Yt(V;\theta_0)\}}{\pi^1(\eta_0^1)}\big\{h(\eta_*,\gamma_*,\theta_*) \big\}^\T \bigg],\quad~~~~~ H=E\bigg\{ \frac{h(\eta_*,\gamma_*,\theta_*)^{\otimes2}}{\pi^1(\eta_0^1)}\bigg\},\\
		G=&E\bigg[\frac{s(W,V;\theta_0)-E\{s(W,V;\theta_0)\}}{1-\pi^1(\eta_0^1)}\big\{ h(\eta_*,\gamma_*,\theta_*)\big\}^\T \bigg],\quad T=E\bigg\{\frac{h(\eta_*,\gamma_*,\theta_*)^{\otimes 2}}{1-\pi^1(\eta_0^1)} \bigg\},
	\end{aligned}
\end{equation*}
where for any matrix $C$, $C^{\otimes 2}=CC^\T$. We further define
\begin{equation*}
	\begin{aligned}
		Q(\eta^1)=&\frac{R}{\pi^1(\eta^1)}\big[Yt(V;\theta_0)-E\{Yt(V;\theta_0)\}\big]-\frac{R-\pi^1(\eta^1)}{\pi^1(\eta^1)}FH^{-1}h(\eta_*,\gamma_*,\theta_*)\\
		&~-\frac{1-R}{1-\pi^1(\eta^1)}\big[s(W,V;\theta_0)-E\{s(W,V;\theta_0)\}\big]+\frac{\pi^1(\eta^1)-R}{1-\pi^1(\eta^1)}GT^{-1}h(\eta_*,\gamma_*,\theta_*),
	\end{aligned}
\end{equation*}
and write $Q=Q(\eta_0^1)$. The asymptotic distribution is given by the following theorem.

\begin{theorem}\label{thm:pi-correct}
	When $\Cone$ contains a correctly specified model for $\pi(V)$, $n^{1/2}(\widehat\theta_{\MR}-\theta_0)$ has an asymptotic normal distribution with mean zero and variance $\var(L)$, where
	$	L=\Gamma^{-1}[Q-E(Q\Psi^\T)\{E(\Psi^{\otimes 2}) \}^{-1}\Psi ]$.
\end{theorem}

It is clear to see that $L$ involves the residual of the projection of $Q$ on $\Psi$, so $\var(L)\leq \var(\Gamma^{-1}Q)$. Note that the latter is the variance of the inverse probability weighting estimator when $\pi(V)$ is known in a data fusion setting.			
This implies that the efficiency of $\widehat\theta_{\MR}$ can be improved by modeling $\pi(V)$ even when it is known. Such a counter-intuitive fact has been studied in traditional missing data problems; see, for example, \cite{robins1995semiparametric} and \cite{han2014multiply}. The following proposition presents the efficient influence function of $\theta_0$ defined through \eqref{eqn:estimatingeq} and provides the corresponding semiparametric efficiency bound. 

\begin{proposition}\label{prop:dr-efficiency}
	The efficient influence function of $\theta_0$ is given by
	\begin{align*}
		\Gamma^{-1}\bigg[\frac{R}{\pi(V)}\big\{Yt(V;\theta_0)-E(Yt(V;\theta_0)\mid V) \big\}
		-\frac{1-R}{1-\pi(V)}\big\{s(W,V;\theta_0)-E(s(W,V;\theta_0)\mid V)\big\}
		\bigg],
	\end{align*}
	and the semiparametric efficiency bound is equal to $\Gamma^{-1}\Omega(\Gamma^{-1})^\T$, where
	\[
	\Omega = E\bigg[\frac{1}{\pi(V)} \var\big\{Yt(V;\theta_0)\mid V\big\}+\frac{1}{1-\pi(V)}\var\big\{s(W,V;\theta_0)\mid V\big\}\bigg].
	\]
\end{proposition}

\cite{chen2008semiparametric} and \cite{shu2020improved} presented similar efficiency bounds in other data fusion settings.
\cite{evans2021doubly} proposed a doubly robust estimator based on the influence function given in Proposition~\ref{prop:dr-efficiency} and mentioned that the estimator is the most efficient  when both the propensity score and imputation models are correctly specified.  We formally state this result here as a supplement and provide more rigorous analysis in the supplementary material. \cite{evans2021doubly} also discussed how to choose a proper function $t(V;\theta)$  to further improve estimation efficiency.

\begin{theorem}\label{thm:efficiency}
	Suppose that  $\Cone$ contains a correctly specified model for $\pi(V)$ and $\Ctwo$ contains a correctly specified model for $f(W\mid V)$. Then  $n^{1/2}(\widehat\theta_{\MR}-\theta_0)$ has an asymptotic normal distribution with mean zero and variance equal to the semiparametric efficiency bound.
\end{theorem}
Different from the local efficiency of the doubly robust estimator proposed by~\cite{evans2021doubly}, the efficiency gain in Theorem~\ref{thm:efficiency} can be achieved without exactly knowing  which two among the multiple models are correctly specified. To make inference, we need to provide a consistent estimator for the asymptotic variance. Such an estimator can be achieved by replacing
the expectations involved in the asymptotic variance with their corresponding sample averages. 


\section{Simulation studies}\label{sec:sim}

In this section, we conduct simulation studies to evaluate the finite sample performance of the proposed procedure. The simulation model has two fully observed covariates $V=(V_1,V_2)^\T$ that are both generated from the standard normal distribution. We generate $W$ from $W\mid V\sim N(-0.5+1.5 V_1+V_2+3V_1*V_2,1)$ and $R$ from the Bernoulli distribution with $R\mid V\sim \Ber\{\pi(V)\}$, where $\pi(V)=\{1+\exp(-0.3+0.75V_1-0.75V_2)\}^{-1}$. The proportion of samples from the primary sample is about 56\%. We finally generate the outcome from $Y\mid W,V\sim N(1+2W+2V_1-1.5V_2,0.4)$. The correct models for $\pi(V)$  and $E(W\mid V)$ are given by $\logit\{\pi^1(\eta^1)\}=\eta_1^1+\eta_2^1V_1+\eta_3^1V_2$ and $a^1(\gamma^1)=\gamma_1^1+\gamma_2^1V_1+\gamma_3^1V_2+\gamma_4^1V_1*V_2$. We also consider the following incorrect working models: $\logit\{\pi^2(\eta^2)\}=\eta_1^2+\eta_2^2V_1$ and $a^2(\gamma^2)=\gamma_1^2+\gamma_2^2V_1+\gamma_3^2V_1*V_2$. The true value of the parameter of interest $\theta^\T=(\theta_1,\theta_2,\theta_3,\theta_4)=(1,2,2,-1.5)$.

We apply the proposed calibration procedure to estimate the parameters and use
the doubly robust  estimator by \cite{evans2021doubly} for comparison. Because the performances of both methods may depend on different combinations of models, we use the four-digit zero-one string to indicate which subset of the working models is used, with 1 denoting use and 0 non-use. The first two digits correspond to the correct and incorrect propensity score models, respectively. The last two digits correspond to the correct and incorrect imputation models, respectively. For example,  $\MR$-1010  represents the proposed estimator using the correct propensity score  and   correct imputation models, whereas $\MR$-1111 indicates that all models are used. For each estimator, we compute the Monte Carlo bias, root mean squared error and 95\% confidence interval coverage probability. The results for sample size $n=500$ and $n=2000$ based on 1000 replications are summarized in Table~\ref{tab:sim}. 

\begin{table}[h!]
	\centering
	\caption{Simulation results based on 1000  replications for $n=500$ and $n=2000$; values have been multiplied by 100.}
	\centering
	\resizebox{0.98\columnwidth}{!}{%
		\begin{tabular}{cccccccccccccccc}
			\hline
			&  \multicolumn{3}{c}{$\theta_1$} & &\multicolumn{3}{c}{$\theta_2$} & & \multicolumn{3}{c}{$\theta_3$} & & \multicolumn{3}{c}{$\theta_4$} \vspace{0.3mm}\\
			\cline{2-4}
			\cline{6-8}
			\cline{10-12}
			\cline{14-16}
			\addlinespace[1mm]
			Estimator & Bias  &RMSE &CP  & &Bias  &RMSE &CP
			& & Bias  &RMSE &CP & &Bias  &RMSE &CP  \\
			\hline
			\addlinespace[1mm]
			& & & & & && $n=500$& &&&& & & &\\
			\addlinespace[0.5mm]
			\hline
			\addlinespace[1mm]
			DR-1010 &$~~~~0$   &~~22 &93~~~~& &~~0  &10 &92.4~~  & &~~~1  &27 &91.9
			& & $-1$  &27 &94.4 \\
			DR-0110 &$~~~~0$   &~~22  &93.3& &~~0  &~~8 &94.8~~ & &~~~1  &28 &93.8
			& & $-1$  &23  &94.1 \\
		DR-1001 &~~~~1   &~~28  &93.7& &$-2$~  &95 &81.8~~ & &~~~3  &71 &93.8
			& & ~~~2  &215~~ &86.6 \\
			DR-0101 &136   &141 & ~~1.3 & &~~0  &18 &89.9~~~& &$-44~~$  &73 &92.3
			& & $-2$  &48 &93.3 \\
			\MR-1010 & $~~~~0$ &~~24  &94.8& &~~0  &~~9 &94.6~~ & & ~~~1 & 29 &93.9
			& & $-1$  & 29 &94.5  \\
			\MR-0110 &$~~~~0$  &~~25 &95~~~~& &~~0  &~~9 &95.2~~  & &~~~1  & 28 &94.3
			& & $-1$  & 28 & 94.6 \\
			\MR-1001 &~~~~0  &~~24  &95.6& &$-1$~  &10 &94.1~~ & &~~~6  & 33 &91.2
			& &~~~0   &34 & 94.9\\
			\MR-0101 &~~55 & ~~61 &50.3& &$-7$~  & 12 &83.5~~ & & ~29 & 42 &81.8
			& &$-6$   & 32  &94.8 \\
			\MR-1110 &$~~~~0$  & ~~24 &95.2  &  & ~~0& ~~9 &94.6~~ & &~~~1 & 29&94~~~ &  & $-1$
			&29 & 94.6 \\
			\MR-1101 	&~~~~0  &~~24  &95.8 &  &$-1$~ &10  &94.1~~  & &~~~5 &31 &91.5 &  & $-1$
			&33 & 94.8 \\
			\MR-1011  & $~~~~0$ & ~~25 &96.2 &  &~~0 &10  &94.9~~& &~~~0 &31 &94.9 &  & $-1$
			&30 & 95.7   \\
			\MR-0111 & $~~~~0$ & ~~25 &96.1   &  &~~0 &12  & 94.8~~& &$~~~0$ &31 &95.4 &  &$-1$ 
			&30 &95.2 \\
			\MR-1111 &~~~~0  & ~~25 & 96.6& &~~0  &10 &95~~~~~& &~~~0  &31 &95.6
			& &$-1$   &30 &95.7  \\
			\hline 
			\addlinespace[1mm]
			& & & & &&& $n=2000$& &&&& &&&\\
			\addlinespace[0.5mm]
			\hline
			\addlinespace[1mm]
			DR-1010 &~~~~0   &~~10 &94.6& &~~0  &~~5 &94.9~~  & &$-1$  &13 &94.9
			& & $~~~0$  &13 & 95.7\\
			DR-0110 &~~~~0   &~~10 &94.8 & &~~0  &~~4 &94.3~~ & &$-1$  &14 &95.5
			& & $~~~0$  &11 &95.2 \\
			DR-1001 &~~~~0   &~~12  &94.8& &~~0  &17 &83.8~~  & &~~~0  &19 &96~~~
			& & $-2$  &48 & 87.4\\
			DR-0101 &135   &136  &0~ & &~~0  &~~8 &93.5~~ & &$-46~~$  &53 &59.2
			& & ~~~1  &24 &92.9\\
			\MR-1010 & ~~~~0 &~~11  & 94.7& &~~0  &~~5 &95.7~~~& & $-1$ & 13 &95.4
			& & $~~~0$  & 14 &95.6  \\
			\MR-0110 &~~~~0  &~~12 &94.1& &~~0  &~~5 &94.9~~  & &$-1$  & 15 &95~~~
			& & $~~~0$  & 15  &94.7 \\
			\MR-1001 &$~~~~0$  &~~11  & 94.1& &$~~0$  &~~5 &94.9~~~& &~~~1  & 15 &93.5
			& &$~~~0$   &17 &95.3 \\
			\MR-0101 & ~~60 & ~~61 & 1~~& &$-7$~  & ~~9 &66.3~~~& & ~26 & 30 &54.2
			& &$-3$   & 17  &94.1\\
			\MR-1110 & ~~~~0 &~~11 &94.7 &  & ~~0&~~5  &95.8~~ & &$-1$ &13  & 95.1&  & $~~~0$
			& 14& 95.2  \\
			\MR-1101 & $~~~~0$ & ~~11&95.1 &  &$~~0$ &~~5  & 95~~~~~ & &~~~1 &14  &94.5 &  & $~~~0$
			&16 &95.1  \\
			\MR-1011 & ~~~~0 & ~~11& 94.9 &  &~~0 & ~~5 &95.8~~ & & $~~~0$& 14 &95.4 &  &$~~~0 $
			&14 & 95.5 \\
			\MR-0111 &~~~~0 &~~11 & 94.9 &  & ~~0&~~5  &95.9~~ & &$~~~0$ &14  &95.6 &  & $~~~0$
			& 14&  95.8 \\
			\MR-1111 &$~~~~0$  & ~~11  &94.8  & &~~0  &~~5 &96~~~~~~& &$~~~0$  &14 &95.6
			& &$~~~0$   &14 &95.6\\
			\hline
		\end{tabular}
	}
	\vspace{1ex}
	
	{\raggedright {\footnotesize RMSE, root mean squared error; CP, coverage probability.} \par}
	\label{tab:sim}
\end{table}

When one  propensity score and one imputation model are used, both the doubly robust estimator  and the calibartion estimator show negligible bias if either model is correctly specified; 
see $\DR$-1010, $\DR$-0110, $\DR$-1001, $\MR$-1010, $\MR$-0110, and $\MR$-1001 in Table~\ref{tab:sim}. However, the peformance of $\DR$-1001 based on the correct propensity score  and incorrect imputation model  is not satisfactory judging from its large root mean squared error. In contrast, 
the proposed calibration estimator $\MR$-1001 is more efficient with smaller root mean squared error. Similar findings are observed when neither model is correct, although both of these two estimators are significantly biased. This  shows that the proposed procedure is less sensitive to  extreme propensity score values and
produces not-too-bad estimates
when both models are misspecified. 
The calibration estimators based on more than two models show ignorable bias and are efficient in almost all scenarios. None of the existing doubly robust estimators can achieve such robustness. Through a calibration strategy, our method effectively accommodates multiple models and delivers more robust and stable estimation.

\section{Application}\label{sec:application}

As an illustration of the proposed method, we consider an assembled data set from the Consumer Expenditure Survey (CEX) and the Survey of Consumer Finance (SCF) to estimate the effects of household asset value on consumption. The CEX is a nationwide annual survey conducted by the U.S. Bureau of Labor Statistics to collect detailed information about household expenditures and some demographic variables. Unfortunately,  the CEX provides limited information on wealth data and thus the CEX alone is not sufficient for our purposes. Previous research \citep{bostic2009housing} has turned to a different triennial survey SCF conducted by Federal Reserve Board to obtain information on U.S. households' assets, liabilities, income and other demographic characteristics. Both the CEX and SCF began in the early 1980s, but for illustration, we focus on the data from CEX's 1997 fourth quarter survey and 1998 SCF, and restrict  the sample with household heads between 25 and 65 years of age. Besides, since the SCF oversamples relatively wealthy households \citep{bostic2009housing}, we truncate the SCF sample at 90$th$ percentiles of observed total household  income and net worth, as was done in  \cite{evans2021doubly}. 
The resulting dataset consists of 5904 households: 3388 from CEX and 2516 from SCF.

Our main interest is the household total net worth (netw) effect on total expenditures (expd) in 1997, adjusting for the total  income before taxes  and certain baseline characteristics of household head, including the continuous covariate age and the binary covariates \sex~($1=$ female), \marital~status ($1=$ married),  education levels: \eduone~($1=$ high school diploma or general educational development); \edutwo~($1=$ some college or Associate degree); \eduthr~($1=$ Bachelors degree or higher), and race types: \white~($1=$ White); \black~($1=$ Black/African American). A logarithmic transformation is required to linearize netw, expd and income, and so the model specification is:
\begin{align*}
	\lexpd=&\theta_1+\theta_2\lnetw+\theta_3\lincome+
	\theta_4\age+\theta_5\sex+\theta_6\marital\\
	&~~~+\theta_7\eduone+\theta_8\edutwo+\theta_9\eduthr+\theta_{10}\white+\theta_{11}\black+\epsilon,
\end{align*}
where $\lexpd$, $\lnetw$, and $\lincome$ denote respectively the logarithmic transformations of netw, expd and income, and $\epsilon$ has mean zero conditional on all covariates. Since the data for $\lexpd$ and $\lnetw$ are collected from two different sources, none of the subjects can simultaneously observe these two variables. It is thus challenging to estimate the effects of household net worth on expenditures while accounting for potential confounders. 

\begin{table}[h!]
	\centering
	\caption{Results analysis for the U.S. household expenditure data. Asterisks denote significance at 0.05 level. }
	\centering
	\begin{tabular}{lcccccccc}
		\hline
		\addlinespace[0.5mm]
		\multirow{ 2}{*}{}  &\multicolumn{2}{c}{CAL-1111} & & \multicolumn{2}{c}{DR-1010}  & & \multicolumn{2}{c}{DR-0101} \vspace{0.1mm}\\
		\cline{2-3}
		\cline{5-6}
		\cline{8-9}
		\addlinespace[0.9mm]
		& Estimate  &SE   &  & Estimate &SE &  & Estimate &SE\\
		\hline
		\addlinespace[1.3mm]
		intercept   &$-1.720^*$~ &0.294 &&$-1.880^*$~ &0.176 &&$-1.899^*$ &0.147  \\
		lnetw  &$~~0.376^*$ &0.096 &&$~~0.359^*$ &0.065 &&$~0.002$ &0.002  \\
		lincome   &~0.043~ &0.098 &&~0.097~ &0.060 &&$~~~0.449^*$ &0.025\\
		age &$-0.156^*$~~ &0.065 &&$-0.157^*$~~ &0.043 &&$~~~0.051^*$ &0.014\\
		sex  &~0.034~ &0.105 & &~0.079~~ &0.064 &&$~0.043$ &0.033  \\
		married    &~0.088~ &0.099 & &~0.127~~ &0.065 &&$~~~0.203^*$ &0.038 \\
		edu1  &~0.077~ &0.126 & &~0.097~~ &0.089 &&$~~~0.142^*$ &0.051  \\
		edu2  &~0.064~ &0.146 & &~0.063~~ &0.114 &&$~~~0.207^*$ &0.055 \\
		edu3  &$-0.038~$~~ &0.157  & &$-0.015~~~~$ &0.144 &&$~~~0.312^*$ &0.056 \\
		white   &$-0.063~~~$ &0.113 & &$-0.072~~~~$ &0.089 &&$~~0.031$ &0.077 \\
		black  &$~0.051$~ &0.149 & &$-0.008~~~~$ &0.115 &&$-0.049~$ &0.088  \\
		\addlinespace[1mm]
		\hline
	\end{tabular}
	\label{tab:realdata}
\end{table}

We apply the proposed method to estimate coefficients in the linear model. 	
To implement our procedure, we consider the following two imputation  models:
a linear regression model $a^1(\gamma^1)$ with  all main effects and quadratic terms for both age and $\lincome$; a second linear regression model $a^2(\gamma^2)$ with all  main effects and an interaction term between age and $\lincome$.
For the propensity score model, we posit two working models that  employ the same
regressors as the imputation regression models in the logit transformation, that is, logit$\{\pi^j(\eta^j)\}=a^j(\gamma^j)$ for $j=1,2$. Because there are missing values in the original survey data, the publicly available data set consists of five imputed replicates. We thus follow \cite{evans2021doubly} to perform estimation for each replicate and combine the results using Rubin's rule \citep{rubin2004multiple}.
We report  the point estimates and standard errors of our analysis results using all the four working models in Table~\ref{tab:realdata}. We also include two classes of the doubly robust estimator results that separately use  the two working models $\{\pi^1(\eta^1),a^1(\gamma^1)\}$ and  the other two working models $\{\pi^2(\eta^2),a^2(\gamma^2)\}$ for comparison.
As indicated by  \cite{evans2021doubly},  $a^1(\gamma^1)$ 	is  nearly a correct specification of  the imputation model. Consequently,  the results from the doubly robust method using $\{\pi^1(\eta^1),a^1(\gamma^1)\}$ are very similar to the proposed calibration method, as shown  in Table~\ref{tab:realdata}.
Based on these results,
households with higher net worth have significantly higher total expenditure, adjusting for household incomes and other covariates.  The covariates income, sex, marital, edu1, edu2, edu3, white and black do not have substantial impacts on household expenditure, whereas age exhibits a significantly negative impact, in that the older household heads have lower total expenditure. These associations generally agree with  previous findings from \cite{bostic2009housing} and~\cite{evans2021doubly}. However, when $\pi^2(\eta^2)$ and $a^2(\gamma^2)$ are chosen to be the working models for the doubly robust estimator, the empirical results are found to be significantly different and may lead to incorrect conclusions. This implies that the doubly robust estimator could suffer from severe bias if both the propensity score and imputation models are unfortunately misspecified. To achieve the double robustness property, practitioners should collect more information about the underlying mechanism
and scrutinize the modelling carefully. The proposed calibration procedure provides a way of improving specifications by incorporating multiple working models. Although incorrect models may be included, 				
the resulting estimator is guaranteed to be consistent as long as there exists one correctly specified model in the estimating procedure.
\section{Discussion}\label{sec:discussion}

We have proposed a calibration approach to regression analysis problems where the  outcome and covariates data are fused from two different sources. 				
Since no subject has complete information in the data fusion setting, existing methods developed for missing data problems cannot be directly applied here. A class of doubly robust estimators based on inverse propensity score weighting is designed particularly for this setting \citep{shu2020improved,evans2021doubly}. Although these estimators offer two chances of achieving consistent estimation, it may be risky to assume either the propensity score or imputation model is correct in practice and they are also sensitive to extreme propensity score values. Our proposed estimator mitigates these issues.
It involves multiple working models that could provide more opportunities to achieve correct specification of the two nuisance models. Since our development builds on the empirical likelihood method that circumvents the use of inverse propensity score, the resulting estimator is not affected dramatically by extreme values. 
Simulation results also demonstrate that the proposed estimator performs not too bad  even if  no model is correctly specified due to the nature of  calibration.

The proposed approach may be improved or extended in several directions. Firstly, the working models are all parametric in this paper, and one can use modern machine learning techniques to further improve robustness. The resulting estimator should still be consistent when one working model is correct, but the rate of convergence may affect the asymptotic distribution. Secondly, it would be interesting to extend the mean regression  to quantile regression problems in data fusion settings. The study of these
issues is beyond the scope of this paper and we leave them as future research topics.

%



					\section*{Appendix}
					
					\begin{proof}[Proof of Proposition~\ref{prop:identifiable}]
						Under Assumption~\ref{ass:ignorability}, we have 
						\begin{equation*}
							f(Y\mid V)= f(Y\mid V,R=1),\quad f(W\mid V)= f(W\mid V,R=0).
						\end{equation*}
						This implies that $f(Y\mid V)$ and $f(W\mid V)$ are  identifiable. Under Assumption~\ref{ass:instrumental-variable}, $Y\indep Z\mid W,X$ and $V=(Z,X^\T)^\T$. Thus, for any given $y$, we have
						\begin{equation*}
							f(y\mid v)=E\{f(y\mid W,x)\mid z,x\}.
						\end{equation*}
						If there exists $f_1(y\mid w,x)$ and $f_2(y\mid w, x)$ satisfying the above equation, then the completeness condition in Assumption~\ref{ass:completeness-condition} implies that $f_1(y\mid w,x)=f_2(y\mid w,x)$; that is, $f(y\mid w,v)$ is identifiable,  and hence, the parameter $\theta_0$ is identifiable.
					\end{proof}
					
					\begin{lemma}\label{lem:weights}
						Suppose that $\pi^1(\eta^1)$ is correctly specified for $\pi(V)$. We have
						\begin{align*}
							\widehat\omega_{1i}=&\frac{1}{n\pi_i^1(\widehat\eta^1)}+o_p(1),\quad ~~~~~~~~~~~(i=1,\ldots,m),\\
							\widehat\omega_{0i}=&\frac{1}{n\{1-\pi_i^1(\widehat\eta^1)\}}+o_p(1),\quad ~~~(i=m+1,\ldots,n).
						\end{align*}
					\end{lemma}
				\begin{proof}
					We first show the results for $\widehat\omega_{1i}$. Based on the main text, the $(J+pK)$-dimensional Lagrange multipliers $\widehat\lambda = (\widehat\lambda_1,\ldots,\widehat\lambda_{J+pK})^\T$ satisfy:
				\begin{equation*}
					\frac{1}{m}\sum_{i=1}^m\frac{\widehat h_{i}(\widehat\eta,\widehat\gamma,\widehat\theta)/\pi_i^1(\widehat\eta^1)}{1+\widehat\lambda^\T \widehat h_{i}(\widehat\eta,\widehat\gamma,\widehat\theta)/\pi_i^1(\widehat\eta^1)}=0,~~ \text{and $1+\widehat\lambda^\T \widehat h_{i}(\widehat\eta,\widehat\gamma,\widehat\theta)/\pi_i^1(\widehat\eta^1)\geq 0$ for $i=1,\ldots,m$.}
				\end{equation*}
				Note that
				\begin{equation}\label{eqn:p-Lagrange-multiplier-constraint}
					\begin{aligned}
						&\frac{1}{m}\sum_{i=1}^m\frac{\widehat h_{i}(\widehat\eta,\widehat\gamma,\widehat\theta)/\pi_i^1(\widehat\eta^1)}{1+\widehat\lambda^\T \widehat h_{i}(\widehat\eta,\widehat\gamma,\widehat\theta)/\pi_i^1(\widehat\eta^1)}\\
						=&\frac{1}{\widehat\tau^1}\frac{1}{m}\sum_{i=1}^m\frac{\widehat h_{i}(\widehat\eta,\widehat\gamma,\widehat\theta)}{1+\frac{\pi_i^1(\widehat\eta^1)-\widehat\tau^1}{\widehat\tau^1}+\big(\frac{\lambda}{\widehat\tau^1} \big)^\T\widehat h_{i}(\widehat\eta,\widehat\gamma,\widehat\theta)}\\
						=&\frac{1}{\widehat\tau^1}\frac{1}{m}\sum_{i=1}^m\frac{\widehat h_{i}(\widehat\eta,\widehat\gamma,\widehat\theta)}{1+\Big(\frac{\lambda_1+1}{\widehat\tau^1},\frac{\lambda_2}{\widehat\tau^1},
							\ldots,\frac{\lambda_{J+pK}}{\widehat\tau^1}\Big)^\T\widehat h_{i}(\widehat\eta,\widehat\gamma,\widehat\theta)}.\\
					\end{aligned}
				\end{equation}
				Comparing this equation with the equation of $\widehat\rho$ in the main text, we have $\widehat\rho_1=(\widehat\lambda_1+1)/\widehat\tau^1$ and $\widehat\rho_l=\widehat\lambda_l/\widehat\tau^1$ for $l=2,\ldots,J+pK$. This implies that
				\begin{equation*}
					\widehat\omega_{1i} = \frac{1}{m}\frac{\widehat\tau^1/\pi_i^1(\widehat\eta^1)}{1+\widehat\lambda^\T \widehat h_{i}(\widehat\eta,\widehat\gamma,\widehat\theta)/\pi_i^1(\widehat\eta^1)}=\frac{\widehat p_i\widehat\tau^1}{\pi_i^1(\widehat\eta^1)}.
				\end{equation*}
			Since 
			$	E\{Rh(\eta_*,\gamma_*,\theta_*)/\pi(V) \}=0$
			and $\eta_*^1=\eta_0^1$, $0$ is the solution to
			\begin{equation*}
				E\bigg\{\frac{Rh(\eta_*,\gamma_*,\theta_*)/\pi^1(\eta_*^1)}{1+\lambda^\T h(\eta_*,\gamma_*,\theta_*)/\pi^1(\eta_*^1)} \bigg\}=0
			\end{equation*}
			as an equation of $\lambda$. Then from the theory of empirical likelihood \citep{owen2001empirical}, $\widehat\lambda=o_p(1)$. Since $\widehat\tau^1-m/n=o_p(1)$, we have $\widehat\omega_{1i}=1/\{n\pi_i^1(\widehat\eta^1)\}+o_p(1)$. 
			
			Next, we show the results for $\widehat\omega_{0i}$. we build the connection between $\widehat\omega_{0i}$ and another empirical likelihood estimator based on the auxiliary sample $\{i:i=m+1,\ldots,n\}$ using the prior knowledge that $\pi^1(\eta^1)$ is correctly specified. Let $q_i$ denote the conditional empirical probability mass on $(Y_i,W_i,V_i)$ given $R_i=0$, $i=m+1,\ldots,n$. Then the estimator of $q_i$ is obtained by solving the following constrained optimization:
			\begin{equation*}\label{eqn:q-constraints}
				\begin{aligned}
					&\max_{q_{m+1},\ldots,q_n}\prod_{i=m+1}^n q_i\quad \text{subject to}~ q_i\geq 0,
					~\sum_{i=m+1}^n q_i \frac{\pi_i^j(\widehat\eta^j)-\widehat\tau^j}{1-\pi_i^1(\widehat\eta^1)}=0\quad (j=1,\ldots,J),\\
					&~~~~~~~~~~~~~~~~~~~~~~~~~~~~~~~~~~~~~~~\sum_{i=m+1}^n q_i\frac{g_{i}(\widehat\gamma^k,\widehat\theta^k)-\widehat\psi^k }{1-\pi_i^1(\widehat\eta^1)}=0\quad (k=1,\ldots,K).
				\end{aligned}
			\end{equation*}
			By the Lagrange multipliers method, we have
			\begin{equation*}
				\widehat q_i = \frac{1}{n-m}\frac{1}{1+\widehat\delta^\T\widehat h_{i}(\widehat\eta,\widehat\gamma,\widehat\theta)/\{1-\pi_i^1(\widehat\eta^1)\}},\quad (i=m+1,\ldots,n),
			\end{equation*}
			where the $(J+pK)$-dimensional Lagrange multipliers $\widehat\delta=(\widehat\delta_1,\ldots,\widehat\delta_{J+pK})^\T$ solves the following equation:
			\begin{equation*}
				\frac{1}{n-m}\sum_{i=m+1}^n\frac{\widehat h_{i}(\widehat\eta,\widehat\gamma,\widehat\theta)/\{1-\pi_i^1(\widehat\eta^1)\}}{1+\widehat\delta^\T \widehat h_{i}(\widehat\eta,\widehat\gamma,\widehat\theta)/\{1-\pi_i^1(\widehat\eta^1)\}}=0,
			\end{equation*}
			and satisfies $1+\widehat\delta^\T \widehat h_{i}(\widehat\eta,\widehat\gamma,\widehat\theta)/\{1-\pi_i^1(\widehat\eta^1)\}\geq 0$ for $i=m+1,\ldots,n$. Similar to the derivation in~\eqref{eqn:p-Lagrange-multiplier-constraint}, we have
			\begin{equation*}
				\begin{aligned}
					&\frac{1}{n-m}\sum_{i=m+1}^n\frac{\widehat h_{i}(\widehat\eta,\widehat\gamma,\widehat\theta)/\{1-\pi_i^1(\widehat\eta^1)\}}{1+\widehat\delta^\T \widehat h_{i}(\widehat\eta,\widehat\gamma,\widehat\theta)/\{1-\pi_i^1(\widehat\eta^1)\}}\\
					=&\frac{1}{1-\widehat\tau^1}\frac{1}{n-m}\sum_{i=m+1}^n\frac{\widehat h_{i}(\widehat\eta,\widehat\gamma,\widehat\theta)}{1+\Big(\frac{\delta_1-1}{1-\widehat\tau^1},\frac{\delta_2}{1-\widehat\tau^1}
						\ldots,\frac{\delta_{J+pK}}{1-\widehat\tau^1}\Big)^\T\widehat h_{i}(\widehat\eta,\widehat\gamma,\widehat\theta)}.
				\end{aligned}
			\end{equation*}
			By comparing this equation with the equation of $\widehat\alpha$ in the main text, we have $\widehat\alpha_1=(\widehat\delta-1)/(1-\widehat\tau^1)$, and $\widehat\alpha_l=\widehat\delta_l/(1-\widehat\tau^1)$ for $l=2,\ldots,J+pK$. Therefore,
			\begin{equation*}
				\widehat\omega_{0i}=\frac{1}{n-m}\frac{(1-\widehat\tau^1)/\{1-\pi_i^1(\widehat\eta^1)\}}{1+\widehat\delta^\T\widehat h_{i}(\widehat\eta,\widehat\gamma,\widehat\theta)/\{1-\pi_i^1(\widehat\eta^1)\}}=\frac{q_i(1-\widehat\tau^1)}{1-\pi_i^1(\widehat\eta^1)},\quad (i=m+1,\ldots,n).
			\end{equation*}
			Since
		$E[(1-R)h(\eta_*,\gamma_*,\theta_*)/\{1-\pi(V)\} ]=0$
		and $\eta_*^1=\eta_0^1$, $0$ is the solution to
		\begin{equation*}
			E\bigg[\frac{(1-R)h(\eta_*,\gamma_*,\theta_*)/\{1-\pi^1(\eta_*^1)\}}{1+\delta^\T h(\eta_*,\gamma_*,\theta_*)/\{1-\pi^1(\eta_*^1)\}} \bigg]=0
		\end{equation*}
		as an equation of $\lambda$. Thus, $\widehat\delta=o_p(1)$. In addition, because $1-\widehat\tau^1- (n-m)/n=o_p(1)$, we then have $\widehat\omega_{0i}=1/[n\{1-\pi_i^1(\widehat\eta^1)\}]+o_p(1)$.

				\end{proof}
					
					\begin{proof}[Proof of Theorem~\ref{thm:consistency-propensity}]
						We aim to show that  $\theta_0$ is the solution to the equation for $\widehat\theta_{\MR}$ as $n\rightarrow\infty$. Then the estimator $\widehat\theta_{\MR}$ is a consistent estimator of $\theta_0$.
						Without loss of generality, we assume $\pi^1(\eta^1)$ is correctly specified for $\pi(V)$.
						As shown in Lemma~\ref{lem:weights}, we have $\widehat\omega_{1i}=1/\{n\pi_i^1(\widehat\eta^1)\}+o_p(1)$, and $\widehat\omega_{0i}=1/[n\{1-\pi_i^1(\widehat\eta^1)\}]+o_p(1)$. Thus,
						\begin{equation}\label{eqn:estimation-equation-omega1}
							\begin{aligned}
								&\Big|\sum_{i=1}^m\widehat\omega_{1i}Y_it(V_i;\theta_0)-E\{Yt(V;\theta_0)\} \Big|\\
								\leq &\bigg|\sum_{i=1}^m\widehat\omega_{1i}Y_it(V_i;\theta_0)-\frac{1}{n}\sum_{i=1}^n\frac{R_i}{
									\pi^1(\widehat\eta^1)}Y_it(V_i;\theta_0) \bigg|\\
								&~+\bigg|\frac{1}{n}\sum_{i=1}^n\frac{R_i}{
									\pi^1(\widehat\eta^1)}Y_it(V_i;\theta_0)-\frac{1}{n}\sum_{i=1}^n\frac{R_i}{
									\pi^1(\eta_0^1)}Y_it(V_i;\theta_0) \bigg|\\
								&~+\bigg| \frac{1}{n}\sum_{i=1}^n\frac{R_i}{
									\pi^1(\eta_0^1)}Y_it(V_i;\theta_0)-E\{Yt(V;\theta_0)\}\bigg|=o_p(1).
							\end{aligned}
						\end{equation}
						All the three terms in the above inequality are equal to $o_p(1)$. Among them,
			the first term holds due to Lemma \ref{lem:weights}. The second term holds due to the consistency of $\widehat{\eta}^1$. Specifically,
			\begin{equation}
			\label{eq: taylor-expan}
			    \begin{aligned}
			        &\bigg|\dfrac1n\sum_{i=1}^n\dfrac{R_i}{\pi^1\left(\widehat\eta^1\right)}Y_it\left(V_i;\theta_0\right)-\dfrac1n\sum_{i=1}^n\dfrac{R_i}{\pi^1\left(\eta_0^1\right)}Y_it\left(V_i;\theta_0\right)\bigg|\\
			        &~~~~~~=~\bigg|\dfrac1n\sum_{i=1}^n\left\{\dfrac1{\pi^1\left(\widehat\eta^1\right)}-\dfrac1{\pi^1\left(\eta_0^1\right)}\right\}Y_iR_iY_it\left(V_i;\theta_0\right)\bigg|\\
			        &~~~~~~\leq~\bigg|\dfrac1n\sum_{i=1}^n\dfrac1{\left\{\pi^1\left(\widehat\eta^\dag\right)\right\}^2}Y_iR_iY_it\left(V_i;\theta_0\right)\bigg|\left|\widehat\eta^1-\eta_0^1\right|=o_p(1),
			    \end{aligned}
			\end{equation}
			where $\eta^\dag$ is an intermediate value between $\widehat{\eta}^1$ and $\eta_0^1$. The third term holds due to the law of large numbers.		Similarly, we have
						\begin{equation*}
							\begin{aligned}
								&\bigg|\sum_{i=m+1}^n \widehat\omega_{0i} s(W_i,V_i;\theta_0)-E\{Yt(V;\theta_0)\}\bigg|\\
								=&~~\bigg|\sum_{i=m+1}^n \widehat\omega_{0i} E\{Yt(V;\theta_0)\mid W_i,V_i;\theta_0\}-E\{Yt(V;\theta_0)\}\bigg|\\
								=&~~\bigg|\sum_{i=m+1}^n \omega_{0i} E\{Yt(V;\theta_0)\mid W_i,V_i;\theta_0\}-\frac{1}{n}\sum_{i=1}^n\frac{1-R_i}{1-\pi_i^1(\widehat\eta^1)}E\{Yt(V;\theta_0)\mid W_i,V_i;\theta_0\}\bigg|\\
								&+\bigg|\frac{1}{n}\sum_{i=1}^n\frac{1-R_i}{1-\pi_i^1(\widehat\eta^1)}E\{Yt(V;\theta_0)\mid W_i,V_i;\theta_0\}-\frac{1}{n}\sum_{i=1}^n\frac{1-R_i}{1-\pi_i^1(\eta_0^1)}E\{Yt(V;\theta_0)\mid W_i,V_i;\theta_0\} \bigg|\\
								&+ \bigg|\frac{1}{n}\sum_{i=1}^n\frac{1-R_i}{1-\pi_i^1(\eta_0^1)}E\{Yt(V;\theta_0)\mid W_i,V_i;\theta_0\}-E\{Yt(V;\theta_0)\} \bigg|=o_p(1).
							\end{aligned}
						\end{equation*}
						Combining this equation with~\eqref{eqn:estimation-equation-omega1} implies that $\theta_0$ is a solution to the equation for $\widehat\theta_{\MR}$ as $n\rightarrow\infty$, and hence, $\widehat\theta_{\MR}$ is a consistent estimator of $\theta_0$ when one of the models in $\Cone$ is correctly specified.
					\end{proof}
					
					\begin{proof}[Proof of Theorem~\ref{thm:consistency-misscov}] 		We aim to show that  $\theta_0$ is the solution to the equation for $\widehat\theta_{\MR}$ again as $n\rightarrow\infty$. Then the estimator $\widehat\theta_{\MR}$ is a consistent estimator of $\theta_0$. 	Without loss of generality,  we assume that $a^1(\gamma^1)$ is correctly specified for $f(W\mid V)$.  Then we have $\widehat\gamma^1\xrightarrow p\gamma^1$ 
				and   $\widehat{\theta}^1\xrightarrow p\theta_0$.  Let $\rho_*$ denote the probability limit of $\widehat\rho$.	
						Note that one of the constraints in~\eqref{eqn:omega1-constraints} is:
						\begin{equation*}
							\begin{aligned}
								\sum_{i=1}^m \widehat \omega_{1i}\bigg[\frac{1}{D}\sum_{d=1}^Ds\{W_i^d(\widehat\gamma^1),V_i;\widehat\theta^1\} \bigg]=\widehat\psi^1 =\frac{1}{n}\sum_{i=1}^n \bigg[\frac{1}{D}\sum_{d=1}^Ds\{W_i^d(\widehat\gamma^1),V_i;\widehat\theta^1\} \bigg].
							\end{aligned}
						\end{equation*}
						Then we have 
		\begin{eqnarray*}\label{eqn:omega1-constraints-missing-cov}
								&&\bigg| \sum_{i=1}^m\widehat\omega_{1i}Y_it(V_i;\theta_0)-E\{Yt(V;\theta_0)\}\bigg|\\
							 &	\leq&\bigg| \sum_{i=1}^m\widehat\omega_{1i}\bigg[Y_it(V_i;\theta_0)-\frac{1}{D}\sum_{d=1}^Ds\{W_i^d(\widehat\gamma^1),V_i;\widehat\theta^1\}\bigg]\bigg| + \bigg|\widehat\psi^1-E\{Yt(V;\theta_0)\} \bigg|\\
							 &	\leq&\bigg|\sum_{i=1}^m\widehat\omega_{1i}\bigg[Y_it(V_i;\theta_0)-\frac{1}{D}\sum_{d=1}^Ds\{W_i^d(\widehat\gamma^1),V_i;\widehat\theta^1\}\bigg]\\
								&&~~~~~~-\frac{1}{m}\frac{1}{1+\rho_{*}^{\T}h(\eta_*,\gamma_*,\theta_*)}\sum_{i=1}^n R_i\bigg[Y_it(V_i;\theta_0)-\frac{1}{D}\sum_{d=1}^Ds\{W_i^d(\widehat\gamma^1),V_i;\widehat\theta^1\}\bigg]\bigg|\\
								&&~+\bigg| \frac{1}{m}\frac{1}{1+\rho_{*}^{\T}h(\eta_*,\gamma_*,\theta_*)}\sum_{i=1}^n R_i\frac{1}{D}\sum_{d=1}^D\Big[s\{W_i^d(\widehat\gamma^1),V_i;\widehat\theta^1\}-s\{W_i^d(\gamma_0^1),V_i;\theta_0\} \Big]		
								\bigg|\\
								&&~+\bigg|\frac{1}{m}\frac{1}{1+\rho_{*}^{\T}h(\eta_*,\gamma_*,\theta_*)}\sum_{i=1}^n R_i\bigg[Y_it(V_i;\theta_0)- \frac{1}{D}\sum_{d=1}^Ds\{W_i^d(\gamma_0^1),V_i;\theta_0\}\bigg] \bigg|\\
								&&~+\bigg| \frac{1}{n}\sum_{i=1}^n\frac{1}{D}\sum_{d=1}^Ds\{W_i^d(\widehat\gamma^1),V_i;\widehat\theta^1\} -\frac{1}{n}\sum_{i=1}^n\frac{1}{D}\sum_{d=1}^Ds\{W_i^d(\gamma_0^1),V_i;\theta_0\} \bigg|\\
								&&~+\bigg| \frac{1}{n}\sum_{i=1}^n\frac{1}{D}\sum_{d=1}^Ds\{W_i^d(\gamma_0^1),V_i;\theta_0^1\}-E\{Yt(V;\theta_0)\}\bigg|=o_p(1).
						\end{eqnarray*}
All terms in the last inequality are equal to $o_p(1)$. Specifically,
the first  term holds because
 \begin{eqnarray*} 
 &	&\bigg|\sum_{i=1}^m\widehat\omega_{1i}\bigg[Y_it(V_i;\theta_0)-\frac{1}{D}\sum_{d=1}^Ds\big\{W_i^d(\widehat\gamma^1),V_i;\widehat\theta^1\big\}\bigg]\\
								&&~~~~~~-\frac{1}{m}\frac{1}{1+\rho_{*}^{\T}h(\eta_*,\gamma_*,\theta_*)}\sum_{i=1}^n R_i\bigg[Y_it(V_i;\theta_0)-\frac{1}{D}\sum_{d=1}^Ds\big\{W_i^d(\widehat\gamma^1),V_i;\widehat\theta^1\big\}\bigg]\bigg| \\ 
								&&	=~~~ \bigg|\frac{1}{m}\sum_{i=1}^nR_i \left\{\frac{1}{1+\widehat\rho^{\T}h(\widehat\eta,\widehat\gamma,\widehat\theta)}-\frac{1}{1+\rho_{*}^{\T}h(\eta_*,\gamma_*,\theta_*)}\right\}\\
					&	 &~~~~~~~~			\times\bigg[Y_it(V_i;\theta_0)-\frac{1}{D}\sum_{d=1}^Ds\big\{W_i^d(\widehat\gamma^1),V_i;\widehat\theta^1\big\}\bigg]\bigg| \\ 
						&	 &=~~~o_p(1).			 \end{eqnarray*} 
						The second  to the fourth   terms  can be similarly proved  as in \eqref{eq: taylor-expan},  and the final term holds due to  the law of large numbers.
						We also note that one of the constraints in~\eqref{eqn:omega0-constraints} is:
						\begin{equation*}
							\begin{aligned}
								&\sum_{i=m+1}^n \widehat \omega_{0i} \bigg[\frac{1}{D}\sum_{d=1}^D s\{W_i^d(\widehat\gamma^1),V_i;\widehat\theta^1\} \bigg]=\frac{1}{n}\sum_{i=1}^n  \bigg[\frac{1}{D}\sum_{d=1}^D s\{W_i^d(\widehat\gamma^1),V_i;\widehat\theta^1\} \bigg].
							\end{aligned}
						\end{equation*}
						Then similar to the above derivation, we can also show that
						\begin{equation*}\label{eqn:omega0-constraints-missing-cov}
							\begin{aligned}
								\bigg| \sum_{i=m+1}^n\widehat\omega_{0i}s(W,V;\theta_0)-E\{Yt(V;\theta_0)\}\bigg|=o_p(1).
							\end{aligned}
						\end{equation*}
						Combining all these equations implies that $\theta_0$ is a solution to the equation for $\widehat\theta_{\MR}$ as $n\rightarrow\infty$. This shows that the proposed estimator $\widehat\theta_{\MR}$ is also a consistent estimator of $\theta_0$ when one of the models in $\Ctwo$ is correctly specified.
					\end{proof}

					To prove Theorem~\ref{thm:pi-correct}, we first present the following lemma. 
					
					\begin{lemma}\label{lem:expression-lambda-delta}
						When $\pi^1(\eta^1)$ is a correctly specified model for $\pi(V)$, we have
						\begin{equation*}
							\begin{aligned}
								\sqrt{n}\widehat\lambda&=H^{-1}\bigg[\frac{1}{\sqrt{n}}\sum_{i=1}^n\frac{R_i-\pi_i^1(\eta_0^1)}{\pi_i^1(\eta_0^1)}h_i(\eta_*,\gamma_*,\theta_*)-\frac{1}{\sqrt{n}}\sum_{i=1}^n A\big\{E(\Psi^{\otimes 2})\big\}^{-1}\Psi_i \bigg]+o_p(1),\\
								\sqrt{n}\widehat\delta&=T^{-1}\bigg[-\frac{1}{\sqrt{n}}\sum_{i=1}^n\frac{R_i-\pi_i^1(\eta_0^1)}{1-\pi_i^1(\eta_0^1)}h_i(\eta_*,\gamma_*,\theta_*)+\frac{1}{\sqrt{n}}\sum_{i=1}^n B\big\{E(\Psi^{\otimes 2})\big\}^{-1}\Psi_i \bigg]+o_p(1).
							\end{aligned}
						\end{equation*}
						where
						\[
						A=E\Bigg[\frac{h(\eta_*,\gamma_*,\theta_*)}{\pi^1(\eta^1_0)}\bigg\{ \frac{\partial\pi^1(\eta_0^1)}{\partial\eta^1}\bigg\}^\T \Bigg],\quad \text{and}\quad 
						B=E\Bigg[\frac{h(\eta_*,\gamma_*,\theta_*)}{1-\pi^1(\eta^1_0)}\bigg\{ \frac{\partial\pi^1(\eta_0^1)}{\partial\eta^1}\bigg\}^\T \Bigg].
						\]
					\end{lemma}
					
					\begin{proof}
						By conditions imposed on the Lagrange multiplies $\widehat{\lambda}$ in the main text, we have
						\[
						0=\frac{1}{n}\sum_{i=1}^n R_i\frac{\widehat h_{i}(\widehat\eta,\widehat\gamma,\widehat\theta)/\pi_i^1(\widehat\eta^1)}{1+\widehat\lambda^\T \widehat h_{i}(\widehat\eta,\widehat\gamma,\widehat\theta)/\pi_i^1(\widehat\eta^1)}.
						\]
						Consequently,
						\begin{eqnarray}
							0&=&\frac{1}{n}\sum_{i=1}^n R_i\frac{\widehat h_{i}(\widehat\eta,\widehat\gamma,\widehat\theta)/\pi_i^1(\widehat\eta^1)}{1+\widehat\lambda^\T \widehat h_{i}(\widehat\eta,\widehat\gamma,\widehat\theta)/\pi_i^1(\widehat\eta^1)}-\frac{1}{n}\sum_{i=1}^n R_i\frac{\widehat h_{i}(\widehat\eta,\widehat\gamma,\widehat\theta)}{\pi_i^1(\widehat\eta^1)} \label{eqn:lambda-tylor}\\
							&&+\frac{1}{n}\sum_{i=1}^n R_i\frac{\widehat h_{i}(\widehat\eta,\widehat\gamma,\widehat\theta)}{\pi_i^1(\widehat\eta^1)}-\frac{1}{n}\sum_{i=1}^n R_i\frac{\widehat h_{i}(\eta_*,\widehat\gamma,\widehat\theta)}{\pi_i^1(\eta_*^1)} \label{eqn:eta-tylor}\\
							&&+\frac{1}{n}\sum_{i=1}^n R_i\frac{\widehat h_{i}(\eta_*,\widehat\gamma,\widehat\theta)}{\pi_i^1(\eta_*^1)}-\frac{1}{n}\sum_{i=1}^n R_i\frac{\widehat h_{i}(\eta_*,\gamma_*,\widehat\theta)}{\pi_i^1(\eta_*^1)} \label{eqn:gamma-tylor}\\
							&&+\frac{1}{n}\sum_{i=1}^n R_i\frac{\widehat h_{i}(\eta_*,\gamma_*,\widehat\theta)}{\pi_i^1(\eta_*^1)}-\frac{1}{n}\sum_{i=1}^n R_i\frac{\widehat h_{i}(\eta_*,\gamma_*,\theta_*)}{\pi_i^1(\eta_*^1)} \label{eqn:theta-tylor}\\
							&&+\frac{1}{n}\sum_{i=1}^n R_i\frac{\widehat h_{i}(\eta_*,\gamma_*,\theta_*)}{\pi_i^1(\eta_*^1)}. \label{eqn:constant-term}
						\end{eqnarray}
						Taking Taylor expansion of the right-hand side of~\eqref{eqn:lambda-tylor} around $\lambda=0$ leads to
						\begin{equation*}
							\begin{aligned}
								\eqref{eqn:lambda-tylor}=-\frac{1}{n}\sum_{i=1}^n R_i\frac{\widehat h_{i}(\widehat\eta,\widehat\gamma,\widehat\theta)^{\otimes 2}}{\{\pi_i^1(\widehat\eta^1)\}^2}\widehat{\lambda}+o_p(n^{-1/2})=-H\widehat{\lambda}+o_p(n^{-1/2}).
							\end{aligned}
						\end{equation*}
						Taking Taylor expansion of \eqref{eqn:eta-tylor} around $\eta=\eta_*$ leads to
						\begin{equation*}
							\begin{aligned}
								\eqref{eqn:eta-tylor}=-\frac{1}{n}\sum_{i=1}^n R_i\frac{\widehat h_{i}(\eta_*,\widehat\gamma,\widehat\theta)}{\{\pi_i^1(\eta_0^1)\}^2}\bigg\{ \frac{\partial\pi_i^1(\eta_0^1)}{\partial\eta^1}\bigg\}^\T(\widehat{\eta}^1-\eta_0^1)+o_p(n^{-1/2})=-A(\widehat{\eta}^1-\eta_0^1)+o_p(n^{-1/2}).
							\end{aligned}
						\end{equation*}
						For expression~\eqref{eqn:gamma-tylor}, one can show that $\{R_i \widehat h_{i}(\eta_*,\gamma,\widehat\theta)/\pi_i^1(\eta_0^1):\Vert \gamma-\gamma_*\Vert\leq \varepsilon\}$ forms a Donsker class and $R_i \widehat h_{i}(\eta_*,\gamma_*,\widehat\theta)/\pi_i^1(\eta_0^1)$ is $L_2$ continuous at $\gamma_*$. Therefore, we have
						\begin{equation*}
							\begin{aligned}
								\eqref{eqn:gamma-tylor}=\frac{1}{n}\sum_{i=1}^n \frac{\partial E\big\{R_i \widehat h_{i}(\eta_*,\gamma_*,\widehat\theta)/\pi_i^1(\eta_0^1) \big\}}{\partial\gamma}(\widehat{\gamma}-\gamma_*)+o_p(n^{-1/2}).
							\end{aligned}
						\end{equation*}
						Similarly, for expression~\eqref{eqn:theta-tylor}, one can show that $\{R_i\widehat h_{i}(\eta_*,\gamma_*,\theta)/\pi_i^1(\eta_0^1):\Vert \theta-\theta_*\Vert\leq \varepsilon\}$ forms a Donsker class and $R_i\widehat h_{i}(\eta_*,\gamma_*,\theta_*)$ is $L_2$ continuous at $\theta_*$. Therefore, we have
						\begin{equation*}
							\begin{aligned}
								\eqref{eqn:theta-tylor}=\frac{1}{n}\sum_{i=1}^n \frac{\partial E\big\{R_i \widehat h_{i}(\eta_*,\gamma_*,\theta_*)/\pi_i^1(\eta_0^1) \big\}}{\partial\theta}(\widehat{\theta}-\theta_*)+o_p(n^{-1/2}).
							\end{aligned}
						\end{equation*}
						It is straightforward to show that both $E\{R_i \widehat h_{i}(\eta_*,\gamma_*,\widehat\theta)/\pi_i^1(\eta_0^1) \}$ and $E\{R_i \widehat h_{i}(\eta_*,\gamma_*,\theta_*)/\pi_i^1(\eta_0^1) \}$ are equal to zero. Hence, both expressions~\eqref{eqn:gamma-tylor} and~\eqref{eqn:theta-tylor} are $o_p(n^{-1/2})$. For expression~\eqref{eqn:constant-term}, it is easy to verify that
						\begin{equation*}
							\begin{aligned}
								\eqref{eqn:constant-term}=\frac{1}{n}\sum_{i=1}^n \frac{R_i-\pi_i^1(\eta_0^1)}{\pi_i^1(\eta_0^1)} h_{i}(\eta_*,\gamma_*,\theta_*)+o_p(n^{-1/2}).
							\end{aligned}
						\end{equation*}
						Combining all the above results yields that
						\begin{equation*}
							\sqrt{n}\widehat{\lambda}=H^{-1}\bigg[\frac{1}{\sqrt{n}}\sum_{i=1}^n\frac{R_i-\pi_i^1(\eta_0^1)}{\pi_i^1(\eta_0^1)}h_i(\eta_*,\gamma_*,\theta_*)-\sqrt{n} A(\widehat{\eta}^1-\eta_0^1) \bigg]+o_p(1).
						\end{equation*}
						In addition, since $\widehat\eta^1$ is obtained via maximum likelihood estimation, we have
						\[
						\sqrt{n}(\widehat{\eta}^1-\eta_0^1)=\frac{1}{\sqrt{n}}\sum_{i=1}^n\big\{E(\Psi^{\otimes 2})\big\}^{-1}\Psi_i+o_p(1).
						\]
						Thus,
						\begin{equation*}
							\sqrt{n}\widehat{\lambda}=H^{-1}\bigg[\frac{1}{\sqrt{n}}\sum_{i=1}^n\frac{R_i-\pi_i^1(\eta_0^1)}{\pi_i^1(\eta_0^1)}h_i(\eta_*,\gamma_*,\theta_*)-\frac{1}{\sqrt{n}}\sum_{i=1}^n A\big\{E(\Psi^{\otimes 2})\big\}^{-1}\Psi_i \bigg]+o_p(1).
						\end{equation*}
						This shows the result for $\widehat{\lambda}$.
						Next, by conditions on the Lagrange multipliers $\widehat{\delta}$ in the proof of Lemma~\ref{lem:weights}, we have
						\begin{equation*}
							0=\frac{1}{n}\sum_{i=1}^n(1-R_i)\frac{\widehat h_{i}(\widehat\eta,\widehat\gamma,\widehat\theta)/\{1-\pi_i^1(\widehat\eta^1)\}}{1+\widehat\delta^\T \widehat h_{i}(\widehat\eta,\widehat\gamma,\widehat\theta)/\{1-\pi_i^1(\widehat\eta^1)\}}.
						\end{equation*}
						Then,
						\begin{eqnarray}
							0&=&\frac{1}{n}\sum_{i=1}^n(1-R_i)\frac{\widehat h_{i}(\widehat\eta,\widehat\gamma,\widehat\theta)/\{1-\pi_i^1(\widehat\eta^1)\}}{1+\widehat\delta^\T \widehat h_{i}(\widehat\eta,\widehat\gamma,\widehat\theta)/\{1-\pi_i^1(\widehat\eta^1)\}}-\frac{1}{n}\sum_{i=1}^n(1-R_i)\frac{\widehat h_{i}(\widehat\eta,\widehat\gamma,\widehat\theta)}{1-\pi_i^1(\widehat\eta^1)}\label{eqn:delta-taylor}\\
							&&+\frac{1}{n}\sum_{i=1}^n(1-R_i)\frac{\widehat h_{i}(\widehat\eta,\widehat\gamma,\widehat\theta)}{1-\pi_i^1(\widehat\eta^1)}-\frac{1}{n}\sum_{i=1}^n(1-R_i)\frac{\widehat h_{i}(\eta_*,\widehat\gamma,\widehat\theta)}{1-\pi_i^1(\eta_0^1)}\label{eqn:delta-eta-taylor}\\
							&&+\frac{1}{n}\sum_{i=1}^n(1-R_i)\frac{\widehat h_{i}(\eta_*,\widehat\gamma,\widehat\theta)}{1-\pi_i^1(\eta_0^1)}-\frac{1}{n}\sum_{i=1}^n(1-R_i)\frac{\widehat h_{i}(\eta_*,\gamma_*,\widehat\theta)}{1-\pi_i^1(\eta_0^1)}\label{eqn:delta-gamma-taylor}\\
							&&+\frac{1}{n}\sum_{i=1}^n(1-R_i)\frac{\widehat h_{i}(\eta_*,\gamma_*,\widehat\theta)}{1-\pi_i^1(\eta_0^1)}-\frac{1}{n}\sum_{i=1}^n(1-R_i)\frac{\widehat h_{i}(\eta_*,\gamma_*,\theta_*)}{1-\pi_i^1(\eta_0^1)}\label{eqn:delta-theta-taylor}\\
							&&+\frac{1}{n}\sum_{i=1}^n(1-R_i)\frac{\widehat h_{i}(\eta_*,\gamma_*,\theta_*)}{1-\pi_i^1(\eta_0^1)}.\label{eqn:delta-constant}
						\end{eqnarray}
						Similar to derivations for expressions~\eqref{eqn:lambda-tylor}--\eqref{eqn:constant-term}, we have
						\begin{eqnarray*}
							\eqref{eqn:delta-taylor}&=&-\frac{1}{n}\sum_{i=1}^n(1-R_i)\frac{\widehat h_{i}(\widehat\eta,\widehat\gamma,\widehat\theta)^{\otimes 2}}{\{1-\pi_i^1(\widehat\eta^1)\}^2}\widehat{\delta}+o_p(n^{-1/2})=-T\widehat\delta +o_p(n^{-1/2}),\\
							\eqref{eqn:delta-eta-taylor}&=&\frac{1}{n}\sum_{i=1}^n(1- R_i)\frac{\widehat h_{i}(\eta_*,\widehat\gamma,\widehat\theta)}{\{1-\pi_i^1(\eta_0^1)\}^2}\bigg\{ \frac{\partial\pi_i^1(\eta_0^1)}{\partial\eta^1}\bigg\}^\T(\widehat{\eta}^1-\eta_0^1)+o_p(n^{-1/2})=B(\widehat{\eta}^1-\eta_0^1)+o_p(n^{-1/2}),\\
							\eqref{eqn:delta-gamma-taylor}&=&\frac{1}{n}\sum_{i=1}^n \frac{\partial E\big[(1-R_i) \widehat h_{i}(\eta_*,\gamma_*,\widehat\theta)/\{1- \pi_i^1(\eta_0^1)\} \big]}{\partial\gamma}(\widehat{\gamma}-\gamma_*)+o_p(n^{-1/2})=o_p(n^{-1/2}),\\
							\eqref{eqn:delta-theta-taylor}&=&\frac{1}{n}\sum_{i=1}^n \frac{\partial E\big[(1-R_i) \widehat h_{i}(\eta_*,\gamma_*,\widehat\theta)/\{1- \pi_i^1(\eta_0^1)\} \big]}{\partial\theta}(\widehat{\theta}-\theta_*)+o_p(n^{-1/2})=o_p(n^{-1/2}),\\
							\eqref{eqn:delta-constant}&=&-\frac{1}{n}\sum_{i=1}^n \frac{R_i-\pi_i^1(\eta_0^1)}{1-\pi_i^1(\eta_0^1)} h_{i}(\eta_*,\gamma_*,\theta_*)+o_p(n^{-1/2}).
						\end{eqnarray*}
						Finally, we obtain that
						\[
						\sqrt{n}\widehat\delta=T^{-1}\bigg[-\frac{1}{\sqrt{n}}\sum_{i=1}^n\frac{R_i-\pi_i^1(\eta_0^1)}{1-\pi_i^1(\eta_0^1)}h_i(\eta_*,\gamma_*,\theta_*)+\frac{1}{\sqrt{n}}\sum_{i=1}^n B\big\{E(\Psi^{\otimes 2})\big\}^{-1}\Psi_i \bigg]+o_p(1).
						\]
						
					\end{proof}
					
					\begin{proof}[Proof of Theorem~\ref{thm:pi-correct}]
						When one of propensity score models is correctly specified, $\widehat\theta_\MR$ converges in probability to $\theta_0$ as $n\rightarrow\infty$. Note that $\sum_{i=1}^m \widehat\omega_{1i}=\sum_{i=m+1}^n\widehat\omega_{0i}=1$ and $E\{Yt(V;\theta_0)\}=E\{s(W,V;\theta_0)\}$. These results combined with the explicit forms of $\widehat\omega_{1i}$ and $\widehat\omega_{0i}$ given in the proof of Lemma~\ref{lem:weights} imply that
						\begin{equation*}
							\begin{aligned}
								o_p(1)=&\frac{\sqrt{n}}{m}\sum_{i=1}^n \frac{R_i\widehat\tau^1/\pi_i^1(\widehat\eta^1)}{1+\widehat\lambda^\T \widehat h_{i}(\widehat\eta,\widehat\gamma,\widehat\theta_{\mathrm{CAL}})/\pi_i^1(\widehat\eta^1)}\big[Y_it(V_i;\widehat\theta_\MR)-E\{Yt(V;\widehat\theta_\MR)\}\big]\\
								&-\frac{\sqrt{n}}{n-m}\sum_{i=1}^n \frac{(1-R_i)(1-\widehat\tau^1)/\{1-\pi_i^1(\widehat\eta^1)\}}{1+\widehat\delta^\T \widehat t_{i}(\widehat\eta,\widehat\gamma,\widehat\theta_{\mathrm{CAL}})/\{1-\pi_i^1(\widehat\eta^1)\}}\big[s(W_i,V_i;\widehat\theta_\MR)-E\{s(W,V;\widehat\theta_\MR)\}\big]
								\\
								\equiv&C_1-C_2.
							\end{aligned}
						\end{equation*}
						We rewrite $C_1$ and $C_2$ respectively as 
						\begin{eqnarray}
							C_1&=&\frac{\sqrt{n}\widehat\tau^1}{m}\sum_{i=1}^n \frac{R_i/\pi_i^1(\widehat\eta^1)}{1+\widehat\lambda^\T \widehat h_{i}(\widehat\eta,\widehat\gamma,\widehat\theta_\MR)/\pi_i^1(\widehat\eta^1)}\big[Y_it(V_i;\widehat\theta_\MR)-E\{Yt(V;\widehat\theta_\MR)\}\big]\label{eqn:asn-lambda0}
							\\
							&&~~~-\frac{\sqrt{n}\widehat\tau^1}{m}\sum_{i=1}^n \frac{R_i}{\pi_i^1(\widehat\eta^1)}\big[Y_it(V_i;\widehat\theta_\MR)-E\{Yt(V;\widehat\theta_\MR)\}\big]\label{eqn:asn-lambda}\\
							&&+\frac{\widehat\tau^1}{m}\sqrt{n}\sum_{i=1}^n \frac{R_i}{\pi_i^1(\widehat\eta^1)}\big[Y_it(V_i;\widehat\theta_\MR)-E\{Yt(V;\widehat\theta_\MR)\}\big]\label{eqn:asn-eta0}\\
							&&~~~-\frac{\widehat\tau^1}{m}\sqrt{n}\sum_{i=1}^n \frac{R_i}{\pi_i^1(\eta_0^1)}\big[Y_it(V_i;\widehat\theta_\MR)-E\{Yt(V;\widehat\theta_\MR)\}\big]\label{eqn:asn-eta}\\
							&&+\frac{\widehat\tau^1}{m}\sqrt{n}\sum_{i=1}^n \frac{R_i}{\pi_i^1(\eta_0^1)}\big[Y_it(V_i;\widehat\theta_\MR)-E\{Yt(V;\widehat\theta_\MR)\}\big]\label{eqn:asn-theta0}\\
							&&~~~-\frac{\widehat\tau^1}{m}\sqrt{n}\sum_{i=1}^n \frac{R_i}{\pi_i^1(\eta_0^1)}\big[Y_it(V_i;\widehat\theta_*)-E\{Yt(V;\widehat\theta_*)\}\big]
							\label{eqn:asn-theta}\\
							&&+\frac{\widehat\tau^1}{m}\sqrt{n}\sum_{i=1}^n \frac{R_i}{\pi_i^1(\eta_0^1)}\big[Y_it(V_i;\widehat\theta_*)-E\{Yt(V;\widehat\theta_*)\}\big],\label{eqn:asn-constant}
						\end{eqnarray}
						and
						\begin{eqnarray}
							C_2&=&\frac{\sqrt{n}(1-\widehat\tau^1)}{n-m}\sum_{i=1}^n \frac{(1-R_i)/\{1-\pi_i^1(\widehat\eta^1)\}}{1+\widehat\delta^\T \widehat h_{i}(\widehat\eta,\widehat\gamma,\widehat\theta)/\{1-\pi_i^1(\widehat\eta^1)\}}\big[s(W_i,V_i;\widehat\theta_\MR)-E\{s(W,V;\widehat\theta_\MR)\}\big]\label{eqn:asn-delta0}\\
							&&~~~~-\frac{\sqrt{n}(1-\widehat\tau^1)}{n-m}\sum_{i=1}^n \frac{1-R_i}{1-\pi_i^1(\widehat\eta^1)}\big[s(W_i,V_i;\widehat\theta_\MR)-E\{s(W,V;\widehat\theta_\MR)\}\big]
							\label{eqn:asn-delta}\\
							&&+\frac{\sqrt{n}(1-\widehat\tau^1)}{n-m}\sum_{i=1}^n \frac{1-R_i}{1-\pi_i^1(\widehat\eta^1)}\big[s(W_i,V_i;\widehat\theta_\MR)-E\{s(W,V;\widehat\theta_\MR)\}\big]\label{eqn:asn-delta-eta0}\\
							&&~~~~-\frac{\sqrt{n}(1-\widehat\tau^1)}{n-m}\sum_{i=1}^n \frac{1-R_i}{1-\pi_i^1(\eta_0^1)}\big[s(W_i,V_i;\widehat\theta_\MR)-E\{s(W,V;\widehat\theta_\MR)\}\big]
							\label{eqn:asn-delta-eta}\\
							&&+\frac{\sqrt{n}(1-\widehat\tau^1)}{n-m}\sum_{i=1}^n \frac{1-R_i}{1-\pi_i^1(\eta_0^1)}\big[s(W_i,V_i;\widehat\theta_\MR)-E\{s(W,V;\widehat\theta_\MR)\}\big]\label{eqn:asn-delta-theta0}\\
							&&~~~~-\frac{\sqrt{n}(1-\widehat\tau^1)}{n-m}\sum_{i=1}^n \frac{1-R_i}{1-\pi_i^1(\eta_0^1)}\big[s(W_i,V_i;\theta_*)-E\{s(W,V;\theta_*)\}\big]\label{eqn:asn-delta-theta}\\
							&&+\frac{\sqrt{n}(1-\widehat\tau^1)}{n-m}\sum_{i=1}^n \frac{1-R_i}{1-\pi_i^1(\eta_0^1)}\big[s(W_i,V_i;\theta_*)-E\{s(W,V;\theta_*)\}\big].\label{eqn:asn-delta-constant}
						\end{eqnarray}
						For expressions~\eqref{eqn:asn-lambda0} and \eqref{eqn:asn-lambda}, taking Taylor expansion around $\lambda=0$ leads to
						\begin{align*}
							\eqref{eqn:asn-lambda0}+\eqref{eqn:asn-lambda}=&-\frac{\widehat\tau^1}{m}\bigg[\sum_{i=1}^n R_i\frac{Y_it(V_i;\widehat\theta_\MR)-E\{Yt(V;\widehat\theta_\MR)\}}{\{\pi_i^1(\widehat\eta^1)\}^2}\big\{\widehat h_{i}(\widehat\eta,\widehat\gamma,\widehat\theta)\big\}^\T \bigg]\sqrt{n}\widehat{\lambda}+o_p(1)\\
							=&-F\sqrt{n}\widehat{\lambda}+o_p(1).
						\end{align*}
						Similarly, for expressions~\eqref{eqn:asn-delta0} and~\eqref{eqn:asn-delta}, taking Taylor expansion around $\delta=0$ leads to
						\begin{equation*}
							\begin{aligned}
								&	\eqref{eqn:asn-delta0}+\eqref{eqn:asn-delta}\\
								=&-\frac{1-\widehat\tau^1}{n-m}\bigg[\sum_{i=1}^n(1-R_i) \frac{s(W_i,V_i;\widehat\theta_\MR)-E\{s(W,V;\widehat\theta_\MR)\}}{\big\{1-\pi_i^1(\widehat\eta^1)\big\}^2}\big\{ \widehat h_{i}(\widehat\eta,\widehat\gamma,\widehat\theta)\big\}^\T\bigg]\sqrt{n}\widehat\delta+o_p(1)\\
								=&-G\sqrt{n}\widehat\delta+o_p(1).
							\end{aligned}
						\end{equation*}
						For the expression $\{\eqref{eqn:asn-eta0}+\eqref{eqn:asn-eta}\}-\{\eqref{eqn:asn-delta-eta0}+\eqref{eqn:asn-delta-eta}\}$, taking Taylor expansion around $\eta=\eta_0^1$ yields that
						\begin{eqnarray*}
								&&\{\eqref{eqn:asn-eta0}+\eqref{eqn:asn-eta}\}-\{\eqref{eqn:asn-delta-eta0}+\eqref{eqn:asn-delta-eta}\}\\
								=&&\frac{1}{\sqrt{n}}\sum_{i=1}^n\bigg[\frac{R_i}{\pi_i^1(\widehat\eta^1)}\Big\{Y_it(V_i;\widehat\theta_\MR)-E\big(Yt(V;\widehat\theta_\MR)\big)\Big\}\\
								&&~~~~~~~~~~~~~-\frac{1-R_i}{1-\pi_i^1(\widehat\eta^1)}\Big\{s(W_i,V_i;\widehat\theta_\MR) -E\big(s(W,V;\widehat\theta_\MR)\big)\Big\}\bigg]\\
								&&~~-\frac{1}{\sqrt{n}}\sum_{i=1}^n\bigg[\frac{R_i}{\pi_i^1(\eta_0^1)}\Big\{Y_it(V_i;\widehat\theta_\MR)-E\big(Yt(V;\widehat\theta_\MR)\big)\Big\}\\
								&&~~~~~~~~~~~~~~~~~~~-\frac{1-R_i}{1-\pi_i^1(\eta_0^1)}\Big\{s(W_i,V_i;\widehat\theta_\MR) -E\big(s(W,V;\widehat\theta_\MR)\big)\Big\}\bigg]+o_p(1)\\
								=&&\Bigg[\frac{1}{n}\sum_{i=1}^n\bigg\{\frac{-R_i}{\pi_i^1(\eta_0^1)^2}\Big(Y_it(V_i;\widehat\theta_\MR)-E\big(Yt(V;\widehat\theta_\MR)\big)\Big)\bigg(\frac{\partial\pi_i^1(\eta_0^1)}{\partial\eta^1} \bigg)^\T\\
								&&~-\frac{1-R_i}{(1-\pi_i^1(\eta_0^1))^2}\Big(s(W_i,V_i;\widehat\theta_\MR)-E\big(s(W,V;\widehat\theta_\MR)\big)\Big)\bigg(\frac{\partial\pi_i^1(\eta_0^1)}{\partial\eta^1} \bigg)^\T \bigg\} \Bigg]\\
								&&~~\times \sqrt{n}(\widehat\eta^1-\eta_0^1)+o_p(1)\\
								=&&-E\Bigg[\bigg\{ \frac{Yt(V;\theta_0)-E(Yt(V;\theta_0))}{\pi(V)}+\frac{s(W,V;\theta_0)-E(s(W,V;\theta_0))}{1-\pi(V)}\bigg\}\bigg\{\frac{\partial\pi^1(\eta_0^1)}{\partial\eta^1} \bigg\}^\T \Bigg]\\
								&&~~\times \sqrt{n}(\widehat\eta^1-\eta_0^1)+o_p(1).
						\end{eqnarray*}
						For the expression $\{\eqref{eqn:asn-theta0}+\eqref{eqn:asn-theta}\}-\{\eqref{eqn:asn-delta-theta0}+\eqref{eqn:asn-delta-theta}\}$, taking Taylor expansion around $\theta=\theta_*=\theta_0$ yields that
						\begin{eqnarray*}
								&&\{\eqref{eqn:asn-theta0}+\eqref{eqn:asn-theta}\}-\{\eqref{eqn:asn-delta-theta0}+\eqref{eqn:asn-delta-theta}\}\\
								=&&\frac{1}{\sqrt{n}}\sum_{i=1}^n\bigg[\frac{R_i}{\pi_i^1(\eta_0^1)}\Big\{Y_it(V_i;\widehat\theta_\MR)-E\big(Yt(V;\widehat\theta_\MR)\big)\Big\}\\
								&&~~~~~~~~~~~-\frac{1-R_i}{1-\pi_i^1(\eta_0^1)}\Big\{s(W_i,V_i;\widehat\theta_\MR)-E\big(s(W,V;\widehat\theta_\MR) \big)\Big\} \bigg]\\
								&&~~-\frac{1}{\sqrt{n}}\sum_{i=1}^n\bigg[\frac{R_i}{\pi_i^1(\eta_0^1)}\Big\{Y_it(V_i;\theta_0)-E\big(Yt(V;\theta_0)\big)\Big\}\\
								&&~~~~~~~~~~~-\frac{1-R_i}{1-\pi_i^1(\eta_0^1)}\Big\{s(W_i,V_i;\widehat\theta_0)-E\big(s(W,V;\widehat\theta_0) \big)\Big\} \bigg]+o_p(1)
								\\=&&\bigg[\frac{1}{n}\sum_{i=1}^n\bigg\{\frac{R_i}{\pi_i^1(\eta_0^1)}\bigg(Y_i\frac{\partial t(V_i;\theta_0)}{\partial\theta}-\frac{\partial E(Yt(V;\theta_0))}{\partial\theta}\bigg)\\
								&&~~~~~~~~~~-\frac{1-R_i}{1-\pi_i^1(\eta_0^1)}\bigg(\frac{\partial s(W_i,V_i;\theta_0)}{\partial\theta}-\frac{\partial E(s(W_i,V_i;\theta_0))}{\partial\theta}\bigg)\bigg\}
								\bigg]\sqrt{n}(\widehat\theta_\MR-\theta_0)+o_p(1)\\
								=&&E\bigg\{Y\frac{\partial t(V;\theta_0)}{\partial\theta}-\frac{\partial s(W_i,V_i;\theta_0)}{\partial\theta} \bigg\}\sqrt{n}(\widehat\theta_\MR-\theta_0)+o_p(1).
						\end{eqnarray*}
						For the expression $\eqref{eqn:asn-constant}-\eqref{eqn:asn-delta-constant}$, we have
						\begin{eqnarray*}
								&&\eqref{eqn:asn-constant}-\eqref{eqn:asn-delta-constant}\\
								=&&\frac{1}{\sqrt{n}}\sum_{i=1}^n\bigg[\frac{R_i}{\pi_i^1(\eta_0^1)}\Big\{Y_it(V_i;\theta_0)-E\big( Yt(V;\theta_0)\big)\Big\}\\
								&&~~~~~~~~~~~~-\frac{1-R_i}{1-\pi_i^1(\eta_0^1)}\Big\{s(W_i,V_i;\theta_0)-E\big( s(W,V;\theta_0)\big)\Big\} \bigg]+o_p(1).
						\end{eqnarray*}
						Combining all the above results yields that
						\begin{eqnarray*}
								o_p(1)=&&C_1-C_2\\
								=&&-F\sqrt{n}\widehat\lambda+G\sqrt{n}\widehat\delta\\
								&&~+\frac{1}{\sqrt{n}}\sum_{i=1}^n\bigg[\frac{R_i}{\pi_i^1(\eta_0^1)}\Big\{Y_it(V_i;\theta_0)-E\big( Yt(V;\theta_0)\big)\Big\}\\
								&&~~~~~~~~~~~~~~~~~-\frac{1-R_i}{1-\pi_i^1(\eta_0^1)}\Big\{s(W_i,V_i;\theta_0)-E\big( s(W,V;\theta_0)\big)\Big\} \bigg]\\
								&&~-E\Bigg[\bigg\{ \frac{Yt(V;\theta_0)-E(Yt(V;\theta_0))}{\pi(V)}+\frac{s(W,V;\theta_0)-E(s(W,V;\theta_0))}{1-\pi(V)}\bigg\}\bigg\{\frac{\partial\pi^1(\eta_0^1)}{\partial\eta^1} \bigg\}^\T \Bigg]\\
								&&~~~~~~~~~\times \sqrt{n}(\widehat\eta^1-\eta_0^1)+E\bigg\{Y\frac{\partial t(V;\theta_0)}{\partial\theta}-\frac{\partial s(W_i,V_i;\theta_0)}{\partial\theta} \bigg\}\sqrt{n}(\widehat\theta_\MR-\theta_0)+o_p(1).
						\end{eqnarray*}
						By Lemma~\ref{lem:expression-lambda-delta} and the expressions of $Q(\eta^1),A,B$, we obtain from the above expression that
						\begin{equation*}
							\begin{aligned}
								0=&\frac{1}{\sqrt{n}}\sum_{i=1}^n Q_i(\eta_0^1)-E\Bigg[\bigg\{\frac{Yt(V;\theta_0)-E(Yt(V;\theta_0))-FH^{-1}h(\eta_*,\gamma_*,\theta_*)}{\pi^1(\eta_0^1)} \\
								&~~~~~~~~+\frac{s(W,V;\theta_0)-E(s(W,V;\theta_0))-GT^{-1}h(\eta_*,\gamma_*,\theta_*)}{1-\pi^1(\eta_0^1)}\bigg\} 
								\bigg\{ \frac{\partial\pi^1(\eta_0^1)}{\partial\eta^1}\bigg\}^\T\Bigg]\sqrt{n}(\widehat\eta^1-\eta_0^1)
								\\&~+E\bigg\{Y\frac{\partial t(V;\theta_0)}{\partial\theta}-\frac{\partial s(W_i,V_i;\theta_0)}{\partial\theta} \bigg\}\sqrt{n}(\widehat\theta_\MR-\theta_0)+o_p(1).
							\end{aligned}
						\end{equation*}
						It is easy to verify that
						\begin{equation*}
							\begin{aligned}
								&E\Bigg[\bigg\{\frac{Yt(V;\theta_0)-E(Yt(V;\theta_0))-FH^{-1}h(\eta_*,\gamma_*,\theta_*)}{\pi^1(\eta_0^1)} \\
								&~~~~~~+\frac{s(W,V;\theta_0)-E(s(W,V;\theta_0))-GT^{-1}h(\eta_*,\gamma_*,\theta_*)}{1-\pi^1(\eta_0^1)}\bigg\} 
								\bigg\{ \frac{\partial\pi^1(\eta_0^1)}{\partial\eta^1}\bigg\}^\T\Bigg]
								=E\big(Q\Psi^\T\big).
							\end{aligned}
						\end{equation*}
						Thus, we have
						\begin{equation*}
							\begin{aligned}
								0=&\frac{1}{\sqrt{n}}\sum_{i=1}^n Q_i(\eta_0^1)-E\big(Q\Psi^\T\big)\sqrt{n}(\widehat\eta^1-\eta_0^1)\\
								&~+E\bigg\{Y\frac{\partial t(V;\theta_0)}{\partial\theta}-\frac{\partial s(W_i,V_i;\theta_0)}{\partial\theta} \bigg\}\sqrt{n}(\widehat\theta_\MR-\theta_0)+o_p(1)\\
								=&\frac{1}{\sqrt{n}}\sum_{i=1}^n \Big[ Q_i(\eta_0^1)-E\big(Q\Psi^\T\big)\big\{E(\Psi^{\otimes 2}) \big\}^{-1}\Psi_i\Big]\\
								&~				+E\bigg\{Y\frac{\partial t(V;\theta_0)}{\partial\theta}-\frac{\partial s(W_i,V_i;\theta_0)}{\partial\theta} \bigg\}\sqrt{n}(\widehat\theta_\MR-\theta_0)+o_p(1).
							\end{aligned}
						\end{equation*}
						This implies that
						\[
						\sqrt{n}(\widehat\theta_\MR-\theta_0)\xrightarrow{d} N\big\{0,\var(Z)\big\},
						\]
						where
						\[
						Z=\Bigg[E\bigg\{Y\frac{\partial t(V;\theta_0)}{\partial\theta}-\frac{\partial s(W,V;\theta_0)}{\partial\theta} \bigg\} \Bigg]^{-1}\Big[Q-E(Q\Psi^\T)\big\{\Psi^{\otimes 2} \big\}^{-1}\Psi \Big].
						\]
						This completes the proof of Theorem~\ref{thm:pi-correct}.
					\end{proof}
					
					\begin{proof}[Proof of Proposition~\ref{prop:dr-efficiency}]
						Consider a parametric path $\beta$ for the joint distribution of $Y$, $W$, $V$ and $R$.
						In order to find the efficient influence function for $\theta_0$, we need to first find a random variable $\Phi:=\Phi(Y,W,V,R)$ with mean 0 and
						\[
						\frac{\partial\theta_0(\beta)}{\partial\beta}\mid_{\beta=0}=E\big\{\Phi S_{\beta}(Y,W,V,R)\big\}\mid_{\beta=0},
						\]
						where $	S_{\beta}(Y,W,V,R)$ is the observed data distribution score, and $\theta_0(\beta)$ is the parameter of interest $\theta_0$ defined through moment condition~\eqref{eqn:estimatingeq} under a regular parametric submodel indexed by $\beta$ that includes the true data generating mechanism at $\beta=0$.
						
						Define $\pi_{\beta}(v)=f_{\beta}(R=1\mid V=v)$. The observed distribution function for $Y$, $W$, $V$ and $R$ is given by
						\begin{align*}\label{eqn:joint-density}
							f_{\beta}(y,w,v,r)=f_{\beta}(v)\{\pi_{\beta}(v)\}^r\{1-\pi_{\beta}(v)\}^{1-r}f_{\beta}(y\mid v)^rf_{\beta}(w\mid v)^{1-r}.
						\end{align*}
						The resulting score function is then given by
						\begin{align*}
							S_{\beta}(y,w,v,r)=(1-r)s_{\beta}(w\mid v)+r s_{\beta}(y\mid v)+\frac{r-\pi_{\beta}(v)}{\pi_{\beta}(v)\{1-\pi_{\beta}(v)\}}\dot{\pi}_{\beta}(v)+s_{\beta}(v),
						\end{align*}
						where 
						\[
						s_{\beta}(w\mid v)=\frac{\partial}{\partial\beta}\log f_{\beta}(w\mid v),\quad s_{\beta}(y\mid v)=\frac{\partial}{\partial \beta}\log f_{\beta}(y\mid v),\quad s_{\beta}(v)=\frac{\partial}{\partial\beta}\log f_{\beta}(v).
						\]
						The tangent space of this model is therefore given by:
						\begin{align}\label{eqn:tangent}
							\mathcal{T}=\big\{(1-r)s_{\beta}(w\mid v)+rs_{\beta}(y\mid v)+a(v)(r-\pi_{\beta}(v))+s_{\beta}(v)\big\},
						\end{align}
						where $\int s_{\beta}(w\mid v)f_{\beta}(w\mid v)dw=0$, $\int s_{\beta}(y\mid v)f_{\beta}(y\mid v)dy=0$, $\int s_{\beta}(v)f_{\beta}(v)dv=0$, and $a(v)$ is any square integrable function. 
						
						Recall that
						\begin{align*}
							E\{Yt(V;\theta_0)-s(W,V;\theta_0)\}=0,
						\end{align*}
						so we have
						\begin{align*}
							\frac{\partial E_{\beta}[Yt\{V;\theta_0(\beta)\}-s\{W,V;\theta_0(\beta)\}]}{\partial\beta}\mid_{\beta=0}=0.
						\end{align*}
						Differentiating under integral gives
						\begin{align*}
							\frac{\partial\theta_0(\beta)}{\partial\beta}\mid_{\beta=0}=-\Gamma^{-1}\Big[E\big\{Yt(V;\theta_0)s_{\beta}(Y\mid V)^\T-s(W,V;\theta_0)s_{\beta}(W\mid V)^\T\big\}_{\mid \beta=0} \Big].
						\end{align*}
						Choose the random variable $\Phi$ to be
						\begin{align*}
							-\Gamma^{-1}\bigg[\frac{R}{\pi(V)}\big\{Yt(V;\theta_0)-E(Yt(V;\theta_0)\mid V) \big\}
							-\frac{1-R}{1-\pi(V)}\big\{s(W,V;\theta_0)-E(s(W,V;\theta_0)\mid V)\big\}\bigg].
						\end{align*}
						It then follows that
						\begin{align*}
							&E\{\Phi\times S_{\beta}(Y,W,X,R)\}_{\mid\beta=0}\\
							=&-\Gamma^{-1}E\bigg[ 
							\frac{R}{\pi(V)}Yt(V;\theta_0)s_{\beta}(Y\mid V)-\frac{R}{\pi(V)}E\{Yt(V;\theta_0)\mid V\}s_{\beta}(Y\mid V)\\
							&~~~~~~~~~~+\frac{R-R\pi(V)}{\pi(V)^2\{1-\pi(V)\}}\big\{Yt(V;\theta_0)-E(Yt(V;\theta_0)\mid V) \big\}\dot{\pi}_{\beta}(V)\\
							&~~~~~~~~~~-\frac{1-R}{1-\pi(V)}s(W,V;\theta_0)s_{\beta}(W\mid V)+\frac{1-R}{1-\pi(V)}E\{s(W,V;\theta_0)\mid V\}s_{\beta}(W\mid V)\\
							&~~~~~~~~~~-\frac{R\pi(V)-R}{\pi(V)\{1-\pi(V)\}^2}\big\{s(W,V;\theta_0)-E(s(W,V;\theta_0)\mid V)\big\}\dot{\pi}_{\beta}(V)	\bigg]_{\mid\beta=0}\\ 	=&-\Gamma^{-1}\Big[E\big\{Yt(V;\theta_0)s_{\beta}(Y\mid V)^\T-s(W,V;\theta_0)s_{\beta}(W\mid V)^\T\big\}_{\mid \beta=0} \Big].
						\end{align*}
						Now one can also verify that $\Phi$ belongs to the tangent space $\mathcal{T}$ in~\eqref{eqn:tangent}, with the first and second terms of $\Phi$ taking the role of $rs_{\beta}(y\mid v)$ and $(1-r)s_{\beta}(w\mid v)$, respectively, and the two other components of \eqref{eqn:tangent} being identically equal to 0. Therefore, $\Phi$ is the efficient influence function of $\theta_0$.
					\end{proof}
					
					\begin{proof}[Proof of Theorem~\ref{thm:efficiency}]
						Define
						\[
						A_1=\frac{R}{\pi^1(\eta_0^1)}\Big[Y t(V;\theta_0)-E\big\{Y t(V;\theta_0) \big\}\Big],\quad B_1=\frac{R}{\pi^1(\eta_0^1)}h(\eta_*,\gamma_*,\theta_*).
						\]
						Then $F=E(A_1B_1^\T)$ and $H=E(B_1^{\otimes 2})$. We first invoke the following two facts:
						\begin{itemize} 
						    \item[(1)] 
						For any function $b(V)$, we have that	$$E\left[\left\{A_1-\frac{\displaystyle R}{\displaystyle\pi^1(\eta_0^1)}\big(E(Yt(V;\theta_0)\mid V)-E(Yt(V;\theta_0))\big)\right\}\left\{\frac{\displaystyle R}{\displaystyle\pi^1(\eta_0^1)}b(V)\right\}\right]=0.$$
						    \item[(2)] 		When $\Ctwo$ contains a correctly specified model for $f(W\mid V)$, $E\{Y t(V;\theta_0)\mid V\}-E\{Y t(V;\theta_0)\}$ is a component of $h(\eta_*,\gamma_*,\theta_*)$. Thus, the vector function $[E\{Y t(V;\theta_0)\mid V\}-E\{Y t(V;\theta_0)\}]R/\pi^1(\eta_0^1)$ is in the linear space spanned by $B_1$.
						\end{itemize} 
	Since  all components of $h(\eta_*,\gamma_*,\theta_*)$ are functions of $V$ only,	 we have that
						\begin{equation*}
							\begin{aligned}
								&FH^{-1}B_1=E(A_1B_1^\T)\big\{E(B_1^{\otimes 2}) \big\}^{-1}B_1
								\\=&E\bigg[\frac{R}{\pi^1(\eta_0^1)}\Big\{E\big(Y t(V;\theta_0)\mid V\big)-E\big(Y t(V;\theta_0) \big)\Big\}B_1^\T \bigg]\big\{E(B_1^{\otimes 2}) \big\}^{-1}B_1\\
								=&\frac{R}{\pi^1(\eta_0^1)}\Big[E\big\{Y t(V;\theta_0)\mid V\big\}-E\big\{Y t(V;\theta_0) \big\}\Big],
							\end{aligned}
						\end{equation*}
		where the first equality holds due to the fact (1) and the second equality follows from the  fact (2).	This shows that $FH^{-1}h(\eta_*,\gamma_*,\theta_*)=E\{Y t(V;\theta_0)\mid V\}-E\{Y t(V;\theta_0) \}$. Similarly, we can show that
						$GT^{-1}h(\eta_*,\gamma_*,\theta_*)=E\{s(W,V;\theta_0)\mid V\}-E\{s(W,V;\theta_0) \}$. Note that $E\{Y t(V;\theta_0) \}=E\{s(W,V;\theta_0) \}$. 
					Then we can simplify the expression of $Q=Q(\eta^1)$ as	\begin{equation*}
						\begin{aligned}
  Q\left(\eta^{1}\right)=& \frac{R}{\pi^{1}\left(\eta^{1}\right)}\left[Y t\left(V ; \theta_{0}\right)-E\left\{Y t\left(V ; \theta_{0}\right)\right\}\right]-\frac{R-\pi^{1}\left(\eta^{1}\right)}{\pi^{1}\left(\eta^{1}\right)} F H^{-1} h\left(\eta_{*}, \gamma_{*}, \theta_{*}\right) \\ &-\frac{1-R}{1-\pi^{1}\left(\eta^{1}\right)}\left[s\left(W, V ; \theta_{0}\right)-E\left\{s\left(W, V ; \theta_{0}\right)\right\}\right]+\frac{\pi^{1}\left(\eta^{1}\right)-R}{1-\pi^{1}\left(\eta^{1}\right)} G T^{-1} h\left(\eta_{*}, \gamma_{*}, \theta_{*}\right) \\
  =&\frac{R}{\pi^{1}\left(\eta^{1}\right)}\big[Y t\left(V ; \theta_{0}\right)-E\left\{Y t\left(V ; \theta_{0}\right) \mid V\right\}\big] \\ &-\frac{1-R}{1-\pi^{1}\left(\eta^{1}\right)}\big[s\left(W, V ; \theta_{0}\right)-E\left\{s\left(W, V ; \theta_{0}\right) \mid V\right\}\big].
\end{aligned}
						\end{equation*}
					A simple calculation yields that $E(Q\Psi^\T)=0$. The desired result then follows from Theorem~\ref{thm:pi-correct}.
					\end{proof}

					\bibliographystyle{apalike}
					\bibliography{mybib}
				\end{document}